\documentclass[a4paper,USenglish,cleveref,autoref,thm-restate]{lipics-v2021} %
\usepackage{mathtools}
\usepackage{xspace}
\usepackage[numbers,sort]{natbib} %
\hypersetup{
    colorlinks = true,
    urlcolor = {black},
    citecolor = {blue},
    linkcolor = {blue}
}

\newcommand{\sv}[1]{}%
 \newcommand{\lv}[1]{#1}%

\usepackage{etoolbox}

 \newcommand{\toappendix}[1]{#1}%
\title{Optimal Path Partitions in Subcubic and Almost-subcubic Graphs}

\author{Tomáš Masařík}
{Institute of Informatics, University of Warsaw, Poland}
{masarik@mimuw.edu.pl}
{0000-0001-8524-4036}{Supported by the Polish National Science Centre SONATA-17 grant number
2021/43/D/ST6/03312.}

\author{Michał Włodarczyk}
{Institute of Informatics, University of Warsaw, Poland}
{m.wlodarczyk@mimuw.edu.pl}
{0000-0003-0968-8414}
{Supported by the Polish National Science Centre SONATA-19 grant number 2023/51/D/ST6/00155.}

\author{Mehmet Akif Yıldız}
{Centrum Wiskunde \& Informatica, Amsterdam, The Netherlands}
{m.akif.yildiz@cwi.nl}
{0000-0002-3349-9165} %
{This work was carried out during the tenure of an ERCIM `Alain Bensoussan' Fellowship Programme.}

\authorrunning{T.~Masařík, M.~Włodarczyk, M.~A.~Yıldız} %

\Copyright{Tomáš Masařík, Michał Włodarczyk, Mehmet Akif Yıldız} %



\category{} %

\relatedversion{} %

\nolinenumbers %

\hideLIPIcs

\numberwithin{equation}{section}

\newcommand{\nn}{\mathbb{N}}
\newcommand{\Dd}[0]{\mathcal{D}}
\newcommand{\Qq}[0]{\mathcal{Q}}
\newcommand{\Pp}[0]{\mathcal{P}}
\newcommand{\Tt}{\mathcal{T}}
\newcommand{\Oh}{\mathcal{O}}
\newcommand{\odd}{\mathsf{odd}}
\newcommand{\pan}{\mathsf{pan}}
\newcommand{\opt}{\mathsf{opt}}
\newcommand{\sen}{\mathsf{sen}}
\newcommand{\high}{\mathsf{high}}
\newcommand{\edges}{\mathsf{edges}}
\newcommand{\pairs}{\mathsf{ends}}
\newcommand{\fo}{\ensuremath{\mathsf{FO}}\xspace}
\newcommand{\fodp}{\ensuremath{\mathsf{FO}\text{+}\mathsf{DP}}\xspace}
\newcommand{\pn}{\mathsf{pn}}
\newcommand{\sub}{\subseteq}
\newcommand{\sm}{\setminus}

\newcommand{\TM}[1]{}
\newcommand{\MAY}[1]{}
\newcommand{\micr}[1]{}
\newcommand{\akif}[1]{}
\newcommand{\mic}[1]{#1}

\usepackage{tikz}
\usepackage{subcaption}
\usepackage{graphicx}

\begin{document}
	
	\maketitle
	\begin{abstract}
We consider the problem of partitioning the edges of a graph into as few paths as possible.
This is a~subject of the classic conjecture of Gallai and
a recurring topic in combinatorics.
Regarding the complexity of partitioning a graph optimally, Peroch\'e [Discret.\ Appl.\ Math.\ 1984] proved that
it is NP-hard already on graphs of maximum degree four, even when we only ask if two paths suffice.

We show that the problem is solvable in polynomial time on subcubic graphs and then we present an efficient algorithm for ``almost-subcubic'' graphs.
Precisely, we prove that the problem is fixed-parameter tractable when parameterized by the edge-deletion distance to a subcubic graph.
To this end, we reduce the task to model checking in first-order logic extended by disjoint-paths predicates (\fodp)
and then we employ the recent tractability result by Schirrmacher, Siebertz, Stamoulis, Thilikos, and Vigny [LICS 2024].
\end{abstract}

\section{Introduction}\label{sec:INTRO}
The famous conjecture of Gallai~\cite{Gallai} states that {the edge set of any} $n$-vertex connected graph can be partitioned into at most $\lceil n/2 \rceil$ paths. 
It has been confirmed on several graphs classes, including planar graphs \cite{Gallai-planar}, graphs of degree at most five~\cite{Bonamy19}, graphs with no even degree vertices \cite{Gallai-2n-over-3-first}, and graphs of treewidth at most four \cite{Gallai-treewidth}.
In general the conjecture remains open and currently the best bound is $\lfloor 2n/3 \rfloor$~\cite{Gallai-2n-over-3-first, Gallai-2n-over-3-second}.
Notably, Lov{\'a}sz  proved that every graph can be decomposed using $\lfloor n/2 \rfloor$ paths and cycles~\cite{Gallai}.

In the recent years, path partitioning and covering problems have 
 also been studied from the algorithmic perspective (see \cite{manuel2018revisiting} for a survey), in particular in
parameterized complexity. %
Depending on the variant, the goal to cover either the vertices or the edges of a given graph
using as few paths as possible.
One may require the paths to be pairwise disjoint and/or to be shortest paths~\cite{Dumas24, baste2025polynomial, chakraborty2025, foucaud2025polynomialtime}.
While the natural parameter is the number of paths~\cite{Dumas24, baste2025polynomial, Fernau25}, the problem has been also studied from the perspective of structural parameterization~\cite{chakraborty2025, foucaud2025polynomialtime, Fernau25} and parameterization below a guarantee~\cite{fomin2025pathcover}.

We consider the problem in the original variant studied by Gallai and for a graph $G$ we 
define $\pn(G)$ as the smallest number of edge-disjoint paths in $G$ whose union is $E(G)$.
Notably, Peroch\'e \cite{NP-complete-paper} showed that deciding whether $\pn(G)\leq 2$ is NP-hard even for graphs $G$ with maximum degree $\Delta(G)=4$.
First, we complement this hardness with a polynomial-time algorithm to compute $\pn(G)$ when $G$ is subcubic, i.e., $\Delta(G) \le 3$.
The sharp transition between  $\Delta(G) = 3$ and $\Delta(G) = 4$ leads to the question of what happens in graphs that are {\em almost subcubic}.
For a graph $G$, we define its {\em subcubic edge-deletion number} -- $\sen(G)$ -- as the smallest number of edges one needs to remove from $G$ to make it subcubic (c.f.~\cite{MajumdarR18}).
We show that determining $\pn(G)$ is fixed-parameter tractable (FPT) with respect to the parameter $\sen(G)$, providing a significant extension of the case  $\Delta(G) = 3$.

\begin{theorem}\label{thm:main}
There is an algorithm that, given a graph $G$, computes $\pn(G)$ in time $f(\sen(G))\cdot |G|^{\Oh(1)}$, where $f$ is a computable function.
\end{theorem}

\subparagraph{Organization of the paper.}
We begin with an informal exposition of the main technical ideas in \Cref{sec:OVERVIEW}.
In \Cref{sec:PRELIMINARIES} we fix notation and discuss basic graph-theoretical notions.
\sv{Due to space constraints, the proof of \Cref{thm:subcubic-case}, covering the subcubic case, is located in Appendix \ref{sec:SUB-CUBIC}.}
\lv{The proof of \Cref{thm:subcubic-case}, covering the subcubic case, is located in \Cref{sec:SUB-CUBIC}.}
The proof of \Cref{thm:main} spans Sections \ref{sec:ALMOST-SUBCUBIC}-\ref{sec:ALGO}.
In particular, in Sections \ref{sec:PATTERNS} and~\ref{sec:MAIN}, we introduce the concept of a pattern and prove its several properties.
Then, in \Cref{sec:ALGO}, we reduce the problem to model checking and conclude the proof.
\sv{The proofs of lemmas marked with ($\bigstar$) are postponed to Appendix \ref{sec:MISSING}.}

\section{Overview}
\label{sec:OVERVIEW}
\subparagraph{Subcubic graphs.}
Let $\odd(G)$ denote the number of odd-degree vertices in $G$.
Clearly, $\pn(G) \ge \odd(G)/2$ because each odd-degree vertex must be an endpoint of some path.
Moreover, %
when all vertices in $G$ have odd degrees then the lower bound is attained (Corollary~\ref{cor:all-degrees-odd}). 
We show that in a subcubic graph the degree-2 vertices can be effectively dissolved, resulting in a convenient formula for $\pn(G)$.
An additional argument is necessary when $G$ contains a {\em pan cycle}, i.e., a cycle with exactly one degree-3 vertex.
Let $\mathrm{pan}(G)$ denote the number of such cycles in $G$.
Then the following formula determines $\pn(G)$ of a subcubic graph $G$ and entails a polynomial algorithm to compute $\pn(G)$. 

\begin{restatable}[Subcubic Path Partition\sv{ $\bigstar$}]{lemma}{thmSubcubic}
\label{thm:subcubic-case}
Let $G$ be a connected graph with $\Delta(G)\leq 3$. Then it holds that 
\begin{align*}
    \pn(G)=\begin{cases}
        2 & \text{if }G\text{ is a cycle or a subdivision of a diamond},\\
        \mathrm{odd}(G)/2 + \mathrm{pan}(G) & \text{otherwise}.
    \end{cases}
\end{align*}
\end{restatable}

\subparagraph{Towards general case: guessing a partial solution.}
We present a brief outline of the algorithm parameterized by $\sen(G)$, starting by an XP algorithm.
Let $V_4 \sub V(G)$ be the set of vertices of degree at least four in $G$.
It is easy to see that $k := \sum_{v \in V_4} \deg(v)$ is of order $\Oh(\sen(G))$.
This also upper bounds the number of edge-disjoint paths that can pass through a vertex from $V_4$.
Let $\Pp_\opt$ denote the family of such paths in some fixed optimal solution.
Since $|\Pp_\opt| \le k$, we can afford guessing some information about each $P\in\Pp_\opt$.
Let $V_X$ denote the set of endpoints of all the paths in $\Pp_\opt$ union with $V_4$.
For each path $P \in \Pp_\opt$, we guess the sequence of vertices from $V_X$ that are visited by $P$---we call this the {\em trace} of $P$ and a collection of traces is called a {\em pattern} (see \Cref{fig:outline:pattern} on Page \pageref{fig:outline:pattern}).
The number of patterns is $f(k) \cdot n^{\Oh(k)}$ and we can enumerate all of them in XP time. 
For a fixed pattern $F$, we can check if there exists a path family $\Pp_F$ matching $F$ (and compute it)
using the known FPT algorithm for finding edge-disjoint paths parameterized by the number of paths~\cite{Kawarabayashi2010}.

Let $F$ be the pattern encoding $\Pp_\opt$ and $\Pp_F$ be the computed family of edge-disjoint paths matching $F$.
Next, let $G_F = G - \Pp_F$ be the graph obtained from $G$ by removing all edges from the paths in $\Pp_F$.
Similarly, let $G_\opt \coloneqq G - \Pp_\opt$; we have $\pn(G) = \pn(G_\opt) + |\Pp_\opt|$.
Since every edge incident to $V_4$ is covered by some $P \in \Pp_F$, the graph  $G_F$ is subcubic and so $\pn(G_F)$ is determined by \Cref{thm:subcubic-case}.
The crux of the analysis is to ensure that $\pn(G_\opt) = \pn(G_F)$ for one pattern $F$, even though the family $\Pp_F$ may differ from $\Pp_\opt$.

The main observation is that $\odd(G_\opt) = \odd(G_F)$ as this depends only on $\odd(G)$ and the pattern $F$.
Hence, the optimal solution must use at least $\odd(G_F)/2$ paths apart from the ones touching $V_4$.
We take advantage of \Cref{thm:subcubic-case} to justify that in fact $\odd(G_F)/2 + |\Pp_F|$ paths suffice to partition $E(G)$.
This becomes challenging when some connected component of $G_F$ is a cycle or a subdivided diamond or it contains a pan cycle.
In our analysis, we effectively rule out these cases by either preprocessing $G$ appropriately or postprocessing the family~$\Pp_F$ by extending some paths (see \Cref{lem:pn-less-odd}).
Consequently, $\pn(G)$ can be computed as the minimum of $\odd(G_F)/2 + |\Pp_F|$ over all patterns $F$ that can be realized using edge-disjoint paths, which yields an XP algorithm.

\lv{\vspace{-0.25em}}

\subparagraph{Incorporating logic.}
Now we explain how to improve the above XP algorithm to an FPT one.
The running time bottleneck comes from guessing the set $V_X$ and so 
we redefine a pattern\footnote{The full definition of a pattern is slightly more technical and also stores information about occurrences of $N_G(V_4)$ on the paths. See \Cref{def:pattern}.} by replacing the vertices from $V_X \setminus V_4$ by variables $x_1,x_2, \dots, x_\ell$ in each trace and additionally storing the degree of each $v \in V_X \setminus V_4$.
The number of such patterns is a function of $k$ only and it is still sufficient to uniquely determine the value of~$\odd(G - \Pp_F)$ for any $\Pp_F$ matching $F$. 

To check if there exists a path family $\Pp_F$ matching $F$,
we seek an assignment of the variables $x_1,x_2, \dots, x_\ell$ in $V(G) \setminus V_4$ that (1) obeys the degree specification 
and so that (2)~the pairs of vertices which occur consecutive in the traces can be connected by 
edge-disjoint paths.
Then we try adjust the computed family $\Pp_F$ to ensure $\pn(G - \Pp_F) = \odd(G - \Pp_F) / 2$.
However, if the assignment places endpoints of one path on a so-called {\em bull cycle} (see \Cref{lem:bull-free}), then $G-\Pp_F$ has a cyclic component that cannot be reduced without increasing $\odd(G - \Pp_F)$.
Hence, we need to add requirement (3) that each resulting path is {\em bull-free}, i.e., it avoids the above configuration.
If all these conditions can be satisfied, we call the pattern $F$ {\em feasible}.
Restricting to bull-free paths allows us to again express $\pn(G)$ as the minimum of $\odd(G_F)/2 + |\Pp_F|$ over all feasible patterns $F$.

The crucial observation is that 
both bull-freeness and the degree specification can be expressed in the first-order logic (\fo) after appropriate preprocessing of the graph.
This allows us to rewrite the question of whether $F$ is feasible using logic: precisely, the \fo logic extended by disjoint-paths predicates (\fodp)~\cite{GolovachST23}.
We utilize the recent result of Schirrmacher, Siebertz, Stamoulis, Thilikos, and Vigny~\cite{Logic} who proved that model checking of \fodp is FPT on bounded-degree graphs\footnote{The theorem in \cite{Logic} is more general and works for every graph class that excludes some graph as a topological minor.}, when parameterized by the formula length.
More precisely, they considered predicates involving paths that are internally vertex-disjoint but we show that this is also sufficient on subcubic graphs.
This allows us to check if a pattern is feasible in time $f(k)\cdot n^{\Oh(1)}$.
Since the number of patterns is a function of $k$, we obtain an FPT algorithm.

\section{Preliminaries}\label{sec:PRELIMINARIES}

Throughout the paper, we use standard graph theory notation and terminology. We always work with simple and undirected graphs. 
We write $H\subseteq G$ to mean $H$ is a \textit{subgraph} of the graph $G$, that is $V(H)\subseteq V(G)$ and $E(H)\subseteq E(G)$. The graph \textit{induced} by $F\subseteq E(G)$, denoted by $G[F]$, is the subgraph of $G$ with vertex set consisting of those vertices incident to edges in $F$ and with edge set $F$. The graph \textit{induced} by $S\subseteq V(G)$, denoted by $G[S]$, is the subgraph of $G$ with the vertex set $S$ and all edges with 
both endpoints in $S$. 

For a collection $\mathcal{Q}$ of graphs, we write $V(\mathcal{Q})=\bigcup_{Q\in \mathcal{Q}}V(Q)$ and $E(\mathcal{Q})=\bigcup_{Q\in \mathcal{Q}}E(Q)$. For a collection $\mathcal{Q}$ of subgraphs of $G$, we write $G-\mathcal{Q}$ for the graph obtained from $G$ by deleting all the edges in $E(\mathcal{Q})$. 
We denote the \textit{degree} of a vertex $v$ by $\deg_G(v)$.
The \textit{maximum degree} of $G$ is $\Delta(G):=\max\{\deg(v):v\in V(G)\}$. We call a graph $G$ \emph{subcubic} if $\Delta(G)\leq 3$. We say $G$ is a \emph{triangle} if $|V(G)|=3=|E(G)|$. We say $G$ is a \emph{diamond} if $|V(G)|=4$ and $|E(G)|=5$.

For a path $P=v_1v_2\ldots v_k$ and $1\leq i<j\leq k$, we write $v_iPv_j$ for the path $v_iv_{i+1}\ldots v_{j-1}v_j$, that is the unique $(v_i,v_j)$-path along $P$.
For an $(a,b)$-path $P$ and $(b,c)$-path $Q$ with $V(P)\cap V(Q)=\{b\}$, we write $aPbQc$ for the $(a,c)$-path that is obtained by \textit{gluing} paths $P$ and $Q$ at the vertex $b$.
We also  combine the two conventions to write concisely $aPbQc$ for the concatenation of the $(a,b)$-subpath of $P$ and the $(b,c)$-subpath of $Q$.

The \textit{length} of path $P$, denoted by $|P|$, is the number of edges in $P$. 
We say a graph $G$ is a \emph{subdivision} of a graph $H$ if $G$ can be obtained from $H$ by successively replacing edges of $H$ with paths.
Two paths $P, Q$ are {\em edge-disjoint} if $E(P) \cap E(Q) = \emptyset$.
Paths $P,Q$ are  {\em internally vertex-disjoint} if no vertex of one path appears as an internal vertex on the other path.

For a graph $G$, a collection of edge-disjoint paths $\Pp$ is said to be a \textit{path partition}\footnote{We remark that it is common in the combinatorial literature to use the name ``path decomposition'' and reserve ``path partition'' for the variant where one aims to partition the set of vertices.
We decided to use ``path partition'' here to avoid confusion with the terminology related to pathwidth.} for $G$ if $\bigcup_{P\in \Pp}E(P)=E(G)$. The \textit{path number} of $G$ is defined as 
	\begin{align*}
		\pn(G):=\min\{r:\text{there exists a path partition }\Pp\text{ of }G\text{ with }|\Pp|=r\}.    
	\end{align*}

We write $\mathrm{odd}(G)$ for the number of odd degree vertices in $G$. We say that a cycle $C$ in $G$ is an \emph{pan cycle} if there exists a unique vertex $v\in V(C)$ with $\deg_G(v)=3$, and $\deg_G(w)=2$ for all other $w\in V(C)$. We write $\mathrm{pan}(G)$ for the number of pan cycles in $G$. Next, we say that a cycle $C$ in $G$ is a \emph{bull cycle} if there are two distinct vertices $u,v\in V(C)$ with $\deg_G(u)=\deg_G(v)=3$, and $\deg_G(w)=2$ for all other $w\in V(C)$.
See \Cref{fig:panBull} for an illustration.

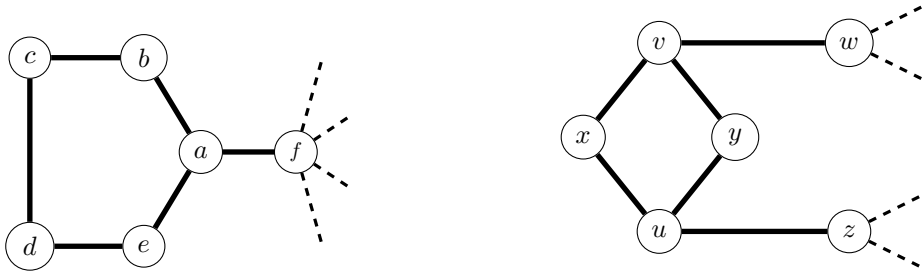
\begin{figure}[!htbp]
    \centering
\begin{subfigure}{.47\textwidth}
\centering
\begin{tikzpicture}[scale=0.5]
\node (v4) at (0,0) [circle, draw] {$d$};
\node (v5) at (3,0) [circle, draw] {$e$};
\node (v1) at (4.5,2.5) [circle, draw] {$a$};
\node (v2) at (3,5) [circle, draw] {$b$};
\node (v3) at (0,5) [circle, draw] {$c$};
\node (v6) at (7,2.5) [circle, draw, scale=0.8] {$f$};

\draw [line width=2pt] (v1)--(v2);
\draw [line width=2pt] (v1)--(v5);
\draw[line width=2pt] (v1)--(v6);
\draw [line width=2pt] (v4)--(v5);
\draw [line width=2pt] (v3)--(v4);
\draw[line width=2pt] (v2)--(v3);

\draw[dashed, line width=1.2pt] (v6)-- (8.5,3.5);
\draw[dashed, line width=1.2pt] (v6)-- (8.5,1.5);
\draw[dashed, line width=1.2pt] (v6)-- (7.7,5);
\draw[dashed, line width=1.2pt] (v6)-- (7.7,0);

\end{tikzpicture}
\end{subfigure}%
\hfill
\begin{subfigure}{.47\textwidth}
\centering
\begin{tikzpicture}[scale=0.5]
\node (w) at (1,2.5) [circle, draw] {$x$};
\node (u) at (3,0) [circle, draw] {$u$};
\node (z) at (5,2.5) [circle, draw] {$y$};
\node (v) at (3,5) [circle, draw] {$v$};
\node (x) at (8,5) [circle, draw] {$w$};
\node (y) at (8,0) [circle, draw] {$z$};

\draw [line width=2pt] (u)--(z);
\draw [line width=2pt] (v)--(z);
\draw[line width=2pt] (v)--(w);
\draw [line width=2pt] (w)--(u);
\draw [line width=2pt] (u)--(y);
\draw[line width=2pt] (x)--(v);

\draw[dashed, line width=1.2pt] (x)-- (10,6);
\draw[dashed, line width=1.2pt] (x)-- (10,4);
\draw[dashed, line width=1.2pt] (y)-- (10,-1);
\draw[dashed, line width=1.2pt] (y)-- (10,1);

\end{tikzpicture}
\end{subfigure}

\caption{An example of a pan cycle $abcde$ ($\deg(a)=3$ and $\deg(b)=\deg(c)=\deg(d)=\deg(e)=2$) and a bull cycle $uyvx$ ($\deg(u)=\deg(v)=3$ and $\deg(x)=\deg(y)=2$).}\label{fig:panBull}
\end{figure}

In any path partition of a graph $G$, each odd-degree vertex of $G$ must be an endpoint of some path. This implies the following.

\begin{observation}\label{obs:lower-bound}
For any graph $G$, it holds that $\pn(G)\geq \odd(G)/2$.
\end{observation}

The proof of an old result of Lov{\'a}sz \mic{implies} that the lower bound is attained when all the degrees are odd. %
\begin{corollary}[\cite{Gallai, Gallai-2n-over-3-first}]\label{cor:all-degrees-odd}
For a graph $G$ with all the degrees odd, we have $\pn(G) = \odd(G)/2$.
\end{corollary}

\toappendix{%
\section{Subcubic Graphs}\label{sec:SUB-CUBIC}

Due to \Cref{cor:all-degrees-odd}, for a subcubic graph $G$ it holds that $\pn(G)=\mathrm{odd}(G)/2$ whenever $\deg_G(v)\in \{1,3\}$ for all $v\in V(G)$. In this section, we extend this result by allowing vertices in $G$ with $d(v)=2$. We begin by showing that pan cycles can be safely removed from the graph. %

\begin{lemma}\label{prop:remove-end-cycles}
Let $G$ be a graph. If $G$ has a pan cycle $C$, then $\pn(G)=\pn(G-C)+1$.    
\end{lemma}
\begin{proof}
Let $C$ be a pan cycle in $G$, and let $x\in V(C)$ be the unique vertex with $\deg_G(x)=3$. Let $y\in V(C)$ be an arbitrary vertex with $y\neq x$. Take a path partition $\Dd$ for $G-C$ with $|\Dd|=\pn(G-C)$. Since $\deg_{G-C}(x)=1$, there exists a $(x,x')$-path $P\in \Dd$ for some $x'\in V(G-C)$, depicted in Figure~\ref{fig:absorb-C}. Observe that $C$ can be decomposed into two $(x,y)$-paths, say $P_1$ and $P_2$. Then, $(\Dd\setminus \{P\})\cup \{x'PxP_1y,yP_2x\}$ is a path partition for $G$ with exactly $\pn(G-C)+1$ paths, which shows that $\pn(G)\leq \pn(G-C)+1$.

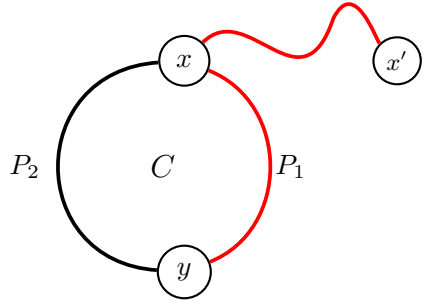
\begin{figure}
    \centering

\begin{tikzpicture}[scale=0.7, thick]

    \node[circle, draw, inner sep=4pt, font=\large] (x) at (0,2) {$x$};
    \node[circle, draw, inner sep=4pt, font=\large] (y) at (0,-2) {$y$};

    \node at (-0.4,0) {\Large $C$};

    \draw[line width=1.4pt,color=red]
        (x) .. controls (2,1.2) and (2,-1.2) .. (y);
    \node at (2,0) {\large $P_1$};

    \draw[line width=1.4pt]
        (x) .. controls (-3,1.8) and (-3,-1.8) .. (y);
    \node at (-3,0) {\large $P_2$};

    \node[circle, draw, inner sep=2.5pt, font=\small] (xp) at (4,2) {$x'$};

    \draw[line width=1.4pt,color=red]
        (x)
        .. controls (1.2,3.2) and (2.2,1.0) ..
        (2.8,2.8)
        .. controls (3.2,3.6) and (3.6,2.4) ..
        (xp);
\end{tikzpicture}
\caption{Extension of the path partition of $G-C$ into a pan cycle $C$.}
\label{fig:absorb-C}
\end{figure}

Conversely, consider a path partition $\Dd$ for $G$ with $|\Dd|=\pn(G)$. Let $z$ be the unique neighbour of $x$ in $G$ with $z\notin V(C)$, and let $P\in \Dd$ be the unique path containing the edge $xz$. Suppose that $P$ is an $(u,w)$-path for some $u\in V(C)$ and $w\in V(G)\setminus V(C)$ (possibly with $u=x$ and $w=z$), such that $x$ is closer to $u$ than $z$, and that $z$ is closer to $w$ than $x$, depicted in Figure~\ref{fig:cut-C}. 

\begin{figure}
    \centering

\begin{tikzpicture}[scale=0.7, thick]

    \tikzstyle{vertex}=[circle, draw, inner sep=4pt, font=\large]

    \node[vertex] (x) at (0,2) {$x$};
    \node[vertex] (y) at (1.8,-1.2) {$y$};
    \node[vertex] (u) at (-1.8,1) {$u$};

    \draw[line width=1.4pt,color=red, dashed]
        (u) .. controls (-0.8,1.95) and (-0.5,1.98) .. (x);

    \draw[line width=1.4pt]
        (x) .. controls (1.2,1.8) and (2.2,0.4) .. (y);

    \draw[line width=1.4pt]
        (y) .. controls (0,-2.6) and (-2.4,-2.0) .. (u);

    \node[circle, draw, inner sep=3pt] (xp) at (4,2) {\footnotesize $w$};
    \node[circle, draw, inner sep=3.2pt, font=\large] (z) at (1.4,3) {$z$};
    \draw [line width=1.4pt,color=red] (x)--(z);

    \draw[line width=1.4pt,color=red]
        (z)
        .. controls (1.8,3.2) and (2.4,1.0) ..
        (2.8,2.8)
        .. controls (3.2,3.6) and (3.6,2.4) ..
        (xp);

     \node at (-0.2,0) {\Large $C$};
     \node at (0.5,3.1) {\Large $P$};
\end{tikzpicture}
\caption{A path partition for $G-C$ via removing the arc $ux$ from the path $P$.}
\label{fig:cut-C}
\end{figure}
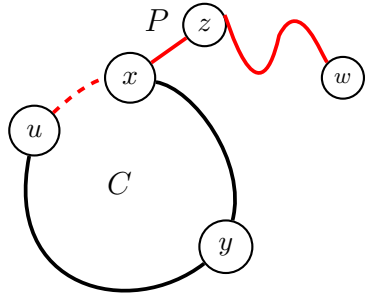

Observe that for all $Q\in \Dd$ with $Q\neq P$, we have either $E(Q)\cap E(C)=\emptyset$ or $E(Q)\subseteq E(C)$. Moreover, there exists $Q_0\in \Dd$ with $Q_0\neq P$ satisfying $E(Q_0)\subseteq E(C)$ as $E(C)\cap E(P)\neq \emptyset$. Then, letting $\Dd_1:=\{Q: Q\in \Dd\text{ and } E(Q)\cap E(C)=\emptyset\}$, we have $\Dd_1\cup \{xzPw\}$ is a path partition for $G-C$. Using $Q_0,P\notin \Dd_1$, we obtain $|\Dd_1|\leq |\Dd|-2=\pn(G)-2$, which implies that $\pn(G-C)\leq \pn(G)-1$, so the result follows.
\end{proof}

We now prove \Cref{thm:subcubic-case}. By \Cref{prop:remove-end-cycles}, it suffices to use $\odd(G)/2$ paths whenever $G$ has no pan cycles, except two edge cases.

\thmSubcubic*
\begin{proof}%
If $G$ is a subdivision of a triangle or a diamond, then the statement is obvious. Assume the otherwise. We will prove the statement by the induction on $|E(G)|$. The claim is clear when $|E(G)|=0$, and assume that it holds for graphs with fewer than $|E(G)|$ edges. If $\mathrm{pan}(G)\geq 1$, then take a pan cycle $C$ in $G$. Note that $\mathrm{odd}(G-C)=\mathrm{odd}(G)$ and that $\mathrm{pan}(G-C)=\mathrm{pan}(G)-1$. Since $G-C$ is connected, by the induction hypothesis, we have $\pn(G-C)=\mathrm{odd}(G-C)/2+\mathrm{pan}(G-C)$. 
The claim follows by \Cref{prop:remove-end-cycles}.

We may assume now that $\mathrm{pan}(G)=0$. By \Cref{cor:all-degrees-odd}, it suffices to find a path partition with $\mathrm{odd}(G)/2$ paths. If $\deg_G(v)\neq 2$ for all $v\in V(G)$, then we have $\mathrm{odd}(G)=|V(G)|$, $\mathrm{even}(G)=0$, and $\mathrm{pan}(G)=0$, so by \Cref{cor:all-degrees-odd}, we are done. Next, take a vertex $v\in V(G)$ with $\deg_G(v)=2$. Let $x$ and $y$ be the neighbors of $v$.

Suppose that $xy\notin E(G)$, and consider $G':=(G-\{vx,vy\})\cup \{xy\}$. Note that $G'$ is connected, $\Delta(G')\leq 3$, and $|E(G')|<|E(G)|$. Moreover, it is easy to see that $G'$ has no pan cycles, and that $\mathrm{odd}(G)=\mathrm{odd}(G')$. By the induction hypothesis, $G'$ has a path partition $\Dd$ with $\mathrm{odd}(G')/2=\mathrm{odd}(G)/2$ paths. Let $P\in \Dd$ be the path containing the edge $xy$. Let $P'$ be the path obtained from $P$ by replacing the edge $xy$ with the path $xvy$. Then, $(\Dd\setminus \{P\})\cup \{P'\}$ is a path partition for $G$ with exactly $\mathrm{odd}(G)/2$ paths, we are done. Hence, we can assume that $xy\in E(G)$. Since $G$ is not a triangle, then we have either $\deg_G(x)=3$ or $\deg_G(y)=3$. Then, by using $\mathrm{pan}(G)=0$, we obtain $\deg_G(x)=\deg_G(y)=3$. Let $G'=G-\{xv,vy,yx\}$.

Suppose that $G'$ is not connected. Then, by using the fact that $G$ is connected, $G'$ has exactly two connected components, say $G_x$ and $G_y$, such that $x\in V(G_x)$ and $y\in V(G_y)$. Note that $\deg_{G_x}(x)=1=\deg_{G_y}(y)$, so neither of $G_x$ and $G_y$ is a subdivision of a triangle or a diamond. Moreover, we have $\mathrm{pan}(G_x)=0=\mathrm{pan}(G_y)$. Hence, by the induction hypothesis, we have a path partition $\Dd_x$ (resp.~$\Dd_y$) for $G_x$ (resp.~$G_y$) with $|\Dd_x|=\mathrm{odd}(G_x)/2$ (resp.~$|\Dd_y|=\mathrm{odd}(G_y)/2$). Since $\deg_{G_x}(x)=1$ (resp.~$\deg_{G_y}(y)=1$), there exists an $(x,x')$-path $P_x\in \Dd_x$ for some $x'\in V(G_x)$ and $(y,y')$-path $P_y\in \Dd_y$ for some $y'\in V(G_y)$, depicted in Figure~\ref{fig:absorp-xvy}. Since $\mathrm{odd}(G_x)/2+\mathrm{odd}(G_y)/2=\mathrm{odd}(G')/2=\mathrm{odd}(G)/2$, we have that $\{x'P_xxvy, y'P_yyx\}\cup (\Dd_x\cup \Dd_y)\setminus \{P_x,P_y\}$ is a path partition for $G$ with exactly $\mathrm{odd}(G)/2$ paths, we are done.

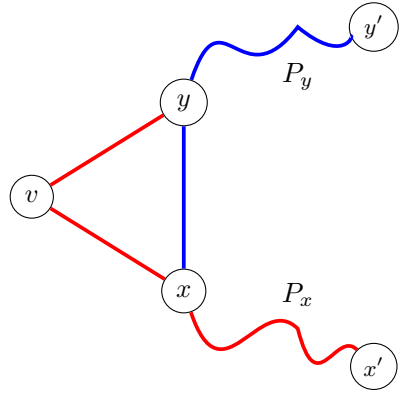
\begin{figure}[!htbp]
    \centering

\begin{tikzpicture}[scale=0.5]

\node (x) at (1,0) [circle, draw] {$x$};
\node (v) at (-3,2.5) [circle, draw] {$v$};
\node (y) at (1,5) [circle, draw] {$y$};

\draw [line width=1.4pt,color=red] (v)--(x);
\draw [line width=1.4pt,color=red] (v)--(y);
\draw[line width=1.4pt,color=blue] (x)--(y);

\node (x') at (6,-2) [circle, draw,inner sep=2.5pt, ] {\footnotesize $x'$};
\node (y') at (6,7) [circle, draw,inner sep=2.5pt, ] {\footnotesize $y'$};

    \draw[line width=1.4pt,color=red]
        (x)
        .. controls (2,-3) and (3,0) ..
        (4,-1)
        .. controls (4.5,-3) and (5,-1) ..
        (x');

    \draw[line width=1.4pt,color=blue]
        (y)
        .. controls (2,8) and (2.5,5) ..
        (4,7)
        .. controls (5.2,6) and (5.5,6.8) ..
        (y');

\node at (4,-0.1) {\large $P_x$};
\node at (4,5.7) {\large $P_y$};

\end{tikzpicture}

\caption{Absorption of the edges $xv,vy,xy$ using $P_x$ and $P_y$.}
\label{fig:absorp-xvy}
\end{figure}

Next, suppose that $G'$ is connected. Then, there exists an $(x,y)$-path $P$ in $G'$. Since $xy\notin E(G')$, the length of $P$ is at least two. If $\deg_{G'}(w)=2$ for all $w\in V(P)\setminus \{x,y\}$, then by using the fact that $G$ is connected, we see that $G'$ itself is $P$, which implies that $G$ is a subdivision of a diamond, a contradiction. Hence, $W:=\{w\in V(P)\setminus\{x,y\}: \deg_{G'}(w)=3\}\neq \emptyset$. Letting $G'':=G'-E(P)$, we have $\deg_{G''}(w)=1$ for all $w\in W$. Moreover, each connected component of $G''$ has at least one vertex from $W$. Therefore, none of the connected components of $G''$ is a subdivision of a triangle or a diamond. Moreover, we have that $\mathrm{odd}(G'')=\mathrm{odd}(G)-2$, and that there are no pan cycles in $G''$. By applying the induction hypothesis for each connected component of $G''$, we find a path partition $\Dd$ for $G''$ with $|\Dd|=\mathrm{odd}(G'')/2=\mathrm{odd}(G)/2-1$. Pick a vertex $w\in W$. Since $\deg_{G''}(w)=1$, there exists a $(w,z)$-path $Q$ in $\Dd$ for some $z\in V(G'')$. If $z\notin V(P)$, then $(\Dd\setminus\{Q\})\cup \{zQwPxy, xvyPw\}$ is a path partition for $G$ with exactly $\mathrm{odd}(G)/2$ paths, as depicted in Figure~\ref{fig:absorp-xvy-when-z-notin-P}, we are done. 

\begin{figure}[!htbp]
    \centering

\begin{tikzpicture}[scale=0.5]

\node (x) at (1,0) [circle, draw] {$x$};
\node (v) at (-3,2.5) [circle, draw] {$v$};
\node (y) at (1,5) [circle, draw] {$y$};

\draw [line width=1.4pt,color=red] (v)--(x);
\draw [line width=1.4pt,color=red] (v)--(y);
\draw[line width=1.4pt,color=blue] (x)--(y);

\node (w) at (6,2) [circle, draw,inner sep=2.5pt] {$w$};

    \draw[line width=1.4pt,color=blue]
        (x)
        .. controls (3,-2) and (5,0) ..
        (w);

    \draw[line width=1.4pt,color=red]
        (y)
        .. controls (3,7) and (5,5) ..
        (w);

\node at (4.2,3.6) {\large $P$};

\node (z) at (14,2.5) [circle, draw,inner sep=2.5pt] {$z$};

\draw[line width=1.4pt,color=blue]
        (w)
        .. controls (9,-2) and (11.5,5) ..
        (z);

\node at (9.5,2.6) {\large $Q$};

\end{tikzpicture}

\caption{Absorption of the edges $xv,vy,xy$ using the paths $P$ and $Q$ when $z\notin V(P)$.}
\label{fig:absorp-xvy-when-z-notin-P}
\end{figure}
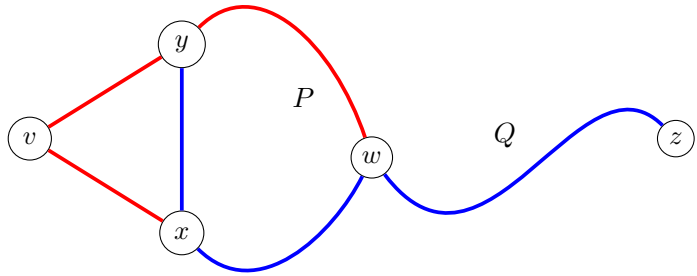

If $z\in V(P)$, without loss of generality, we can assume that $z$ is between $x$ and $w$ along $P$, depicted in Figure~\ref{fig:absorp-xvy-when-z-in-P}. Hence, $(\Dd\setminus\{Q\})\cup \{wQzPxvy, zPyx\}$ is a path partition for $G$ with exactly $\mathrm{odd}(G)/2$ paths, so the result follows.

\begin{figure}[!htbp]
    \centering

\begin{tikzpicture}[scale=0.5]

\node (x) at (1,0) [circle, draw] {$x$};
\node (v) at (-3,2.5) [circle, draw] {$v$};
\node (y) at (1,5) [circle, draw] {$y$};

\draw [line width=1.4pt,color=blue] (v)--(x);
\draw [line width=1.4pt,color=blue] (v)--(y);
\draw[line width=1.4pt,color=red] (x)--(y);

\node (w) at (6,2) [circle, draw,inner sep=2.5pt] {$w$};

    \draw[line width=1.4pt,color=red]
        (y)
        .. controls (3,7) and (5,5) ..
        (w);

\node (z) at (4.6,-1.1) [circle, draw,inner sep=2.5pt] {$z$};

        \draw[line width=1.4pt,color=blue]
        (x)
        .. controls (2.6,-1.7) and (2.6,-1.7) ..
        (z);

        \draw[line width=1.4pt,color=red]
        (z)
        .. controls (5.6,0) and (5.6,0) ..
        (w);

\node at (4.2,3.6) {\large $P$};

\draw[line width=1.4pt,color=blue]
        (w)
        .. controls (14,2) and (12,0) ..
        (10,-2)
        .. controls (9,-2.5) and (7.5,-2) ..
        (z);

\node at (9.5,2.6) {\large $Q$};

\end{tikzpicture}

\caption{Absorption of the edges $xv,vy,xy$ using the paths $P$ and $Q$ when $z\in V(P)$.}
\label{fig:absorp-xvy-when-z-in-P}
\end{figure}
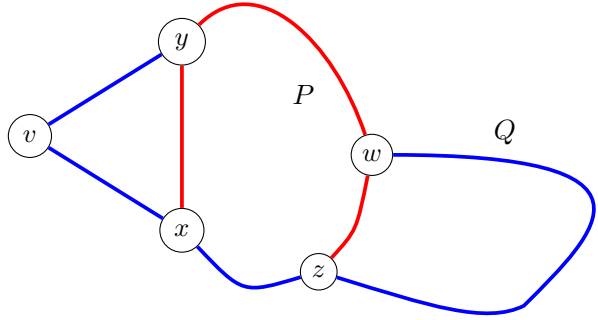

\end{proof}
}%

\section{Almost-Subcubic Graphs} %
\label{sec:ALMOST-SUBCUBIC}

We begin with several observations that will come in useful in the proof of \Cref{thm:main}.
First we preprocess the input to avoid inconvenient corner cases.
In particular, we discard the pan cycles and shorten the bull cycles.
The first facilitates application of \Cref{thm:subcubic-case} and the latter allows us to detect bull cycles in an \fo formula.

\begin{definition}
We say that graph is \emph{nice} if it is: (a) connected, (b) not subcubic, (c) does not have any {pan cycles}, and (d) all bull cycles are triangles.

A {\em nice pair} $(G,V_4)$ is given by a nice graph $G$ and $V_4 \subseteq V(G)$ being the set of vertices of degree at least four in $G$.
\end{definition}

The following lemma justifies that we can restrict ourselves to nice graphs only. 

\begin{restatable}{lemma}{lemNiceGraphs}
\label{lem:niceGraphs}
Let $G$ be a connected graph that is not subcubic.
One can transform $G$ in polynomial time into a nice graph $G_0$
so that $\pn(G)=\pn(G_0)+\mathrm{pan}(G)$ and $\sen(G)=\sen(G_0)$.
\end{restatable}

\lv{%

\begin{proof}
Let $C$ be a pan cycle in $G$, and $x\in V(C)$ be the unique vertex satisfying $\deg_G(x)=3$. Let $G'$ be the graph obtained from $G$ by deleting all the vertices $V(C)\setminus \{x\}$, i.e., $G'$ is the graph obtained from $G-C$ by deleting the isolated vertices. Note that $G'$ is connected. By \Cref{prop:remove-end-cycles}, we have $\pn(G)=\pn(G')+1$. 
\mic{Moreover, all the affected vertices have degree at most 3, hence this modification does not change $\sen(G)$.}
Since $\mathrm{pan}(G')=\mathrm{pan}(G)-1$, by successively removing all the pan cycles in the similar way, we end up with a connected and not subcubic graph $G_1$ without pan cycles such that $\pn(G)=\pn(G_1)+\mathrm{pan}(G)$ and $\sen(G_1)= \sen(G)$.
Next, let $B$ be a bull cycle in $G_1$ that is not a triangle. Let $u,v\in V(B)$ be the vertices satisfying $\deg_{G_1}(u)=\deg_{G_1}(v)=3$. Take an arbitrary vertex $w\in V(B)\setminus\{u,v\}$. Let $B_{uv}$ be the $(u,v)$-path on $B$ not containing $w$. Similarly define $B_{uw}$ and $B_{vw}$. Let $G_1'$ be the graph obtained from $G_1$ by deleting all the vertices $V(B)\setminus \{u,v,w\}$ and adding edges $uv$, $vw$, $uw$ (if they are not already present). Note that $G_1'$ is connected. Moreover, if $\deg_{G_1}(z)\geq 4$ for some $z\in V(G_1)$, then we have $\deg_{G_1}(z)=\deg_{G_1'}(z)$, so $G_1'$ is not subcubic and $\sen(G_1')=\sen(G_1)$. For any path partition of $G_1'$ of minimum size, we can obtain a path partition of $G_1$ of the same size by replacing the edges $uv$, $vw$, $uw$ with the paths $B_{uv}$, $B_{vw}$, $B_{uw}$, respectively, which implies that $\pn(G_1)\leq \pn(G_1')$. Conversely, let $\mathcal{P}$ be a path partition for $G_1$ of size $\pn(G_1)$. 

Suppose that there exists $P\in \Pp$ that is not entirely contained in $B$ such that $P-B$ is not a single path in $G_1'$. Since $\deg_{G_1}(z)=2$ for all $z\in V(B)\setminus \{u,v\}$, we see that $P$ should contain either $B_{uv}$ or $B_{uw}\cup B_{wv}$ in the middle with non-empty subpaths at both ends, depicted in Figure~\ref{fig:P-B-not-single-path}. 

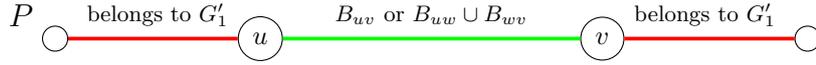
\begin{figure}[!htbp]
\centering
\begin{tikzpicture}[scale=0.45]
\node (u) at (0,0) [circle, draw] {$u$};
\node (v) at (10,0) [circle, draw] {$v$};
\node (x) at (-6,0) [circle, draw] {};
\node (y) at (16,0) [circle, draw] {};

\node at (5,0.7) [color=black] {\footnotesize $B_{uv}$ or $B_{uw}\cup B_{wv}$};

\node at (-7,0.7) [color=black] {\Large $P$};

\node at (-3,0.7) [color=black] {\footnotesize belongs to $G_1'$};

\node at (13,0.7) [color=black] {\footnotesize belongs to $G_1'$};

\draw[line width=1.4pt, color=green] (u)--(v);

\draw[line width=1.4pt,color=red] (u)-- (x);
\draw[line width=1.4pt,color=red] (v)-- (y);

\end{tikzpicture}

\caption{If $P-B$ is not a single path in $G_1'$, then there is an $(u,v)$-subpath in the middle.}
\label{fig:P-B-not-single-path}
\end{figure}

This, in particular, implies that there exists at least one path $Q\in \Pp$ that is entirely contained in $B$. Moreover, since $\deg_{G_1'}(u)=\deg_{G_1'}(v)=1$, there are no other paths $R\in \Pp$ for which $R-B$ is neither empty nor a single path. As a result, together with two paths obtained from $P-B$, we obtain a path partition for $G_1'-\{uv,vw,uw\}$ with at most $|\Pp|-2=\pn(G_1)-2$ paths by deleting edges belonging to $B$ from every path in $\Pp\setminus \{P,Q\}$. Since $u$ and $v$ lie on different paths, we can extend these paths with $uwv$ to the $u$-end and $vu$ to the $v$-end in order to obtain a path partition for $G_1'$ of size $\pn(G_1)$, depicted in Figure~\ref{fig:P-B-not-single-path-extended}, which shows $\pn(G_1')\leq \pn(G_1)$.

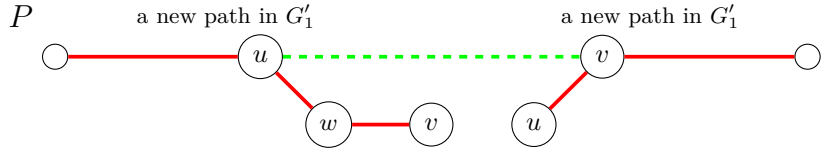
\begin{figure}[!htbp]
\centering
\begin{tikzpicture}[scale=0.45]

\node (u1) at (0,-4) [circle, draw] {$u$};
\node (v1) at (10,-4) [circle, draw] {$v$};
\node (x1) at (-6,-4) [circle, draw] {};
\node (y1) at (16,-4) [circle, draw] {};

\node (u2) at (8,-6) [circle, draw] {$u$};
\node (v2) at (5,-6) [circle, draw] {$v$};
\node (w) at (2,-6) [circle, draw] {\small $w$};

\node at (-1,-2.8) [color=black] {\footnotesize a new path in $G_1'$};

\node at (-7,-2.8) [color=black] {\Large $P$};

\node at (11.4,-2.8) [color=black] {\footnotesize a new path in $G_1'$};

\draw[line width=1.4pt,color=green, dashed] (u1)-- (v1);

\draw[line width=1.4pt,color=red] (u1)-- (x1);
\draw[line width=1.4pt,color=red] (v1)-- (y1);
\draw[line width=1.4pt,color=red] (u1)-- (w);
\draw[line width=1.4pt,color=red] (w)-- (v2);
\draw[line width=1.4pt,color=red] (v1)-- (u2);

\end{tikzpicture}
\caption{Subpaths belonging to $G_1'$ in $P$ can be extended in order to include edges $uv$, $vw$, $uw$.}
\label{fig:P-B-not-single-path-extended}
\end{figure}

Finally, suppose that for every $P\in \Pp$, either $P$ is entirely contained in $B$ or $P-B$ is a single path in $G_1'$. Then, for every path in $\Pp$ that contains at least one edge from $B$, one of the following three cases can happen, depicted in Figure~\ref{fig:types-having-edge-from-B}.

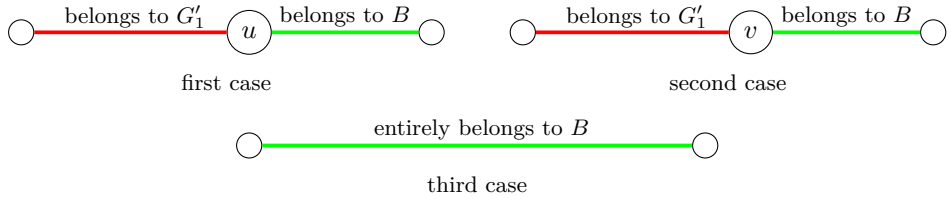
\begin{figure}[!htbp]
\centering
\begin{tikzpicture}[scale=0.3]
\node (u) at (0,0) [circle, draw] {$u$};
\node (a) at (-10,0) [circle, draw] {};
\node (b) at (8,0) [circle, draw] {};

\node at (4.2,0.7) [color=black] {\footnotesize belongs to $B$};

\node at (-5,0.7) [color=black] {\footnotesize belongs to $G_1'$};

\draw[line width=1.4pt,color=red] (u)-- (a);
\draw[line width=1.4pt,color=green] (u)-- (b);

\node at (-1,-2.2) [color=black] {\footnotesize first case};

\node (u1) at (22,0) [circle, draw] {$v$};
\node (a1) at (12,0) [circle, draw] {};
\node (b1) at (30,0) [circle, draw] {};

\node at (26.2,0.7) [color=black] {\footnotesize belongs to $B$};

\node at (17,0.7) [color=black] {\footnotesize belongs to $G_1'$};

\draw[line width=1.4pt,color=red] (u1)-- (a1);
\draw[line width=1.4pt,color=green] (u1)-- (b1);

\node at (21,-2.2) [color=black] {\footnotesize second case};

\node (a2) at (0,-5) [circle, draw] {};
\node (b2) at (20,-5) [circle, draw] {};

\node at (10.2,-4.3) [color=black] {\footnotesize entirely belongs to $B$};

\draw[line width=1.4pt,color=green] (a2)-- (b2);

\node at (10,-6.7) [color=black] {\footnotesize third case};

\end{tikzpicture}
\caption{There are three cases for every path in $\Pp$ that has at least one edge from $B$.}
\label{fig:types-having-edge-from-B}
\end{figure}

If there are at least two paths in $\Pp$ that satisfies the third case, then we obtain a path partition for $G_1'-\{uv,vw,uw\}$ with at most $|\Pp|-2=\pn(G_1)-2$ paths by deleting edges belonging to $B$ from every path in $\Pp$. Then, by adding new paths $uwv$ and $uv$, we find a path partition for $G_1'$ with at most $\pn(G_1)$ paths, so $\pn(G_1')\leq \pn(G_1)$ holds. If there is exactly one path $P\in \Pp$ that satisfies the third case, then without loss of generality, we can assume that there is a path $Q\in \Pp$ that satisfies the first case. Hence, we obtain a path partition for $G_1'-\{uv,vw,uw\}$ with $|\Pp|-1=\pn(G_1)-1$ paths by deleting edges belonging to $B$ from every path in $\Pp\setminus \{P\}$. Then, we obtain a path partition for $G_1'$ of size $\pn(G_1)$ by adding a new path $vuw$ and extending $Q$ by adding the edge $uv$ to the $u$-end, which shows $\pn(G_1')\leq \pn(G_1)$. If there are no paths in $\Pp$ satisfying the third case, using $\deg_{G_1}(u)=\deg_{G_1}(v)=3$, we can conclude that there is exactly one path $P_u\in \Pp$ satisfying the first case and exactly one path $P_v\in \Pp$ satisfying the second case. Then, $B$ is entirely contained in $P_u\cup P_v$, which implies that the part of $P_u$ belonging to $B$ ends at $v$, and that the part of $P_v$ belonging to $B$ ends at $u$. Then, $(\Pp\setminus \{P_u,P_v\})\cup \{P_u-B,P_v-B\}$ is a path partition for $G_1'-\{uv,vw,uw\}$ with $|\Pp|=\pn(G_1)$ paths, where $v\notin P_u-B$ and $u\notin P_v-B$. Hence, we obtain a path partition for $G_1'$ of size $\pn(G_1)$ by extending $P_u-B$ with the path $uwv$ to the $u$-end and extending $P_v-B$ with the edge $vu$ to the $v$-end, which shows $\pn(G_1')\leq \pn(G_1)$. We proved $\pn(G_1')\leq \pn(G_1)$ holds in all cases, so we have $\pn(G_1')=\pn(G_1)$.

Note that $G_1'$ has one less bull cycle that is not a triangle compared to $G_1$. Then, by successively reducing such bull cycles in the similar way, we end up with a connected and not subcubic graph $G_0$ whose all bull cycles are triangles such that $\pn(G_1)=\pn(G_0)$ and $\sen(G_1)=\sen(G_0)$. Using $\pn(G)=\pn(G_1)+\mathrm{pan}(G)$ and $\sen(G_1)=\sen(G)$, the result follows.
\end{proof}
}

Instead of working with the parameter $\sen(G)$ directly, we consider 
the number of edges touching the high-degree vertices, that is,
 $\high(G) = \sum_{v \in V_4} \deg_G(v)$.

\begin{lemma}\label{lem:high-sen}
    Let $(G,V_4)$ be a nice pair.
    Then $\high(G) \le 8 \cdot \sen(G)$.
\end{lemma}
\begin{proof}
    Let $S \sub E(G)$ be a minimum-size set for which $G - S$ is subcubic.
    For $v \in V_4$ let $\deg_S(v)$ denote the number of edges from $S$ incident to $v$.
    Observe that $\deg_S(v) \ge \deg_G(v) - 3$ and $\deg_G(v) \ge 4$ so $\deg_G(v) \le 4\cdot \deg_S(v)$.
    Since each edge from $S$ touches at most two vertices from $V_4$ we have $\sum_{v \in V_4} \deg_S(v) \le 2\cdot |S|$.
    The lemma follows.
\end{proof}

The special treatment of bull cycles is related to the following property.

\begin{definition}[Bull-free Paths]
    A {\em bull triangle} is a bull cycle of length three.
    The two degree-3 vertices of a bull triangle form a {\em bull-pair}.
    A path is {\em bull-free} if its endpoints do not form a bull-pair.
    A family of paths $\Pp$ is bull-free if every $P \in \Pp$ is bull-free. 
\end{definition}

The next lemma allows us to assume that in an optimal path partition, every path visiting $V_4$ is bull-free.

\begin{lemma}\label{lem:bull-free}
    Let $(G,V_4)$ be a nice pair, $\Pp$ be a path partition for $G$ of minimum size, and $\Pp_4 \sub \Pp$ be the subfamily of paths that visit at least one vertex from $V_4$.
    Then $\Pp_4$ is bull-free. 
\end{lemma}
\begin{proof}
Suppose that $\Pp_4$ is not bull-free and let $P \in \Pp_4$ be the path with endpoints $u,v$ that are the degree-3 vertices of a bull triangle $T=(u,v,w)$.
    Because $V(P) \cap V_4 \ne \emptyset$, the path $P$ cannot be fully contained in $T$.
    This implies that $w \not\in V(P)$ and $uv\notin E(P)$. %
    Then, the triangle $T$ forms a connected component of the graph $G \setminus E(P)$.
    Hence, $\Pp$ must contain two paths $Q_1,Q_2$ covering $T$.
    But appending $w$ to the $v$-end of $P$ yields a valid path $P'$, depicted in Figure~\ref{fig:bull-free-paths}. 
    Then $\Pp':=\Pp \setminus \{P,Q_1,Q_2\} \cup \{P', wuv\}$ is a path partition for $G$ with $|\Pp'|<|\Pp|$, a~contradiction.
\end{proof}

\begin{figure}[h]
\centering
\includegraphics[width=0.8\textwidth]{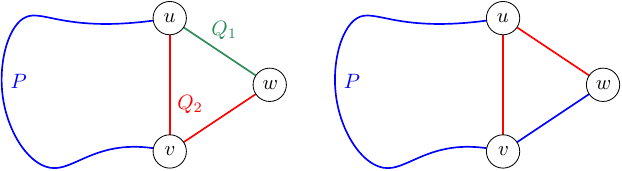}
\caption{Extending the path $P$ with the edge $vw$ %
reduces the number of used paths by one.
}
\label{fig:bull-free-paths}
\end{figure}

\section{Patterns} %
\label{sec:PATTERNS}

\subsection{Traces and encoding}

Let $(G,V_4)$ be a nice pair.
For a path partition $\Pp$ of $G$ let $\Pp_4 \sub \Pp$ be the subfamily of paths that visit at least one vertex from $V_4$.
It is easy to see that $|\Pp_4| \le \high(G)$ so we can afford guessing some information about each path $P \in \Pp_4$.
Specifically, we will define a {\em pattern} encoding the endpoints of each  $P \in \Pp_4$ and its visits in $N_G[V_4]$.
However, we cannot afford guessing all the endpoints, so we will 
encode the relevant vertices outside $N_G[V_4]$ as variables and record their degrees.
To make the exposition clearer, we first define a {\em trace}, being a representation of a single path, then we
explain how the traces are extracted from $\Pp_4$, and finally we provide a full definition of a pattern. 

\begin{definition}[Trace]
    Let $(G,V_4)$ be a nice pair and $\ell$ be an integer. 
    We define $X_\ell \coloneqq \{x_1, \dots, x_\ell\}$. %
    The elements of $X_\ell$ are called {\em variables}.
    
    A {\em $(G,V_4,\ell)$-trace} is a sequence over the alphabet $N_G[V_4] \cup X_\ell$ in which no symbol occurs more than once and no variable $x_i$ appears next to any $v \in V_4$.
\end{definition}

The last condition reflects the fact that whenever a path $P$ visits a vertex $v \in V_4$ then the consecutive vertices on $P$ must belong to $N_G(v) \sub N_G[V_4]$ while the variables $x_1, \dots, x_\ell$ are meant to encode vertices outside $N_G[V_4]$.
See \Cref{fig:outline:pattern} for an example.

In the later arguments (\Cref{lem:patern-conditions}) it will be helpful to avoid the same pair of consecutive symbols appearing twice in the traces.
The following structure allows us to impose this property.

\begin{definition}[Terminal Collection]\label{def:terminal}
    Let $(G,V_4)$ be a nice pair and $\Pp$ be a family of edge-disjoint paths in $G$.
    A set of vertices $U \sub V(G)$ is called a {\em terminal collection} for $\Pp$
    if
    \begin{enumerate}
        \item $N_G[V_4] \sub U$,
        \item for each $P \in \Pp$ the endpoints of $P$ belong to $U$,
        \item for each $u,v \in U$, if two paths $P_1,P_2 \in \Pp$ visit both $u$ and $v$ then at least one of $P_1,P_2$ visits another vertex $w \in U$ between $u$ and $v$.
    \end{enumerate}
\end{definition}

\begin{lemma}\label{lem:terminal} %
    Let $(G,V_4)$ be a nice pair, $\Pp$ be a path partition of $G$, and $\Pp_4 \subseteq \Pp$ be the subfamily of paths that visit any vertex from $V_4$.
    Then $\Pp_4$ admits a terminal collection of size at most $16\cdot\high(G)$.
\end{lemma}
\begin{proof}
    We first set $U_0$ as the union of $N_G[V_4]$ and all the endpoints of paths in $\Pp_4$.
    We have $|N_G[V_4]| \le 2\cdot \high(G)$ and $|\Pp_4| \le \high(G)$ hence $|U_0| \le 4\cdot\high(G)$.
    Next, let $U = N_G[U_0]$.
    Observe that for each $v \in U_0$ we have $|N_G(v) \sm U_0| \le 3$ because if $\deg_G(v) > 3$ then $v \in V_4$ and $N_G(v) \sub U_0$.
    Therefore, we can estimate $|U| \le 4\cdot |U_0| \le 16\cdot \high(G)$.

    We claim that $U$ forms a {terminal collection} for $\Pp_4$.
    Suppose not and consider $u,v \in U$ and $P_1,P_2 \in \Pp_4$ so that each of $P_1,P_2$ visits $u,v$ but no vertex from $U$ is visited by them in between.
    Since $P_1,P_2$ are edge-disjoint, at least one of $P_1,P_2$ visits at least one vertex from $V(G)$ between $u,v$.
    However, we assumed that no such vertex belongs to $U$, which implies that $u,v \not \in U_0$.
    Consequently, none of $u,v$ forms an endpoint of $P_1$ or $P_2$.
    As $u,v$ are internal vertices in both $P_1,P_2$, we infer that $\deg_G(u) \ge 4$ and $\deg_G(v) \ge 4$ hence $u,v \in V_4 \sub T_0$.
    The obtained contradiction concludes the proof.
\end{proof}

To facilitate working with the subfamily $\Pp_4 \sub \Pp$, we define a structure that encapsulates its properties.

\begin{definition}[Covering family]\label{def:cover}
Let $(G,V_4)$ be a nice pair and $\Qq$ be a family of edge-disjoint paths in $G$.
We say that $\Qq$ is a {\em covering family} if every edge $e \in E(G)$ incident to $V_4$ is contained in some path from $\Qq$.
\end{definition}

\begin{figure}[h]
\centering
\includegraphics[width=0.75\textwidth]{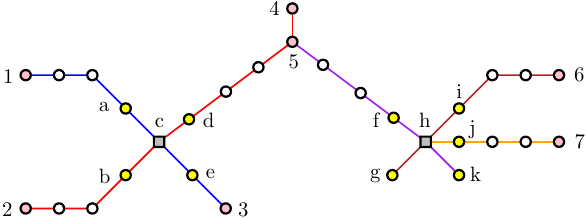}
\caption{The vertices in $V_4$ ($c$ and $h$) are gray squares, the ones in $N(V_4)$ are yellow disks labeled with the remaining letters, and the vertices in $V_X$ are pink disks labeled with numbers.
The highlighted vertices form a terminal collection for the covering family comprising the five paths represented by
the edge colors.
The pattern contains the following traces: red $(x_2,b,c,d,x_5,x_4)$, blue $(x_1,a,c,e,x_3)$,
purple $(x_5,f,h,k)$, brown $(g,h,i,x_6)$, orange $(h,j,x_7)$.
}
\label{fig:outline:pattern}
\end{figure}

\begin{definition}[Encoding]\label{def:match}
    Let $(G,V_4)$ be a nice pair and $\Qq$ be a covering family. %
    Next, let $\Tt$ be a collection of $(G,V_4,\ell)$-traces for some $\ell \in \nn$ and $d$ be a function  $X_\ell \to \{1,2,3\}$.
    We say that the pair $(\Tt, d)$ {\em encodes} $\Qq$ if it can be obtained in the following way.

    Let $U$ be some terminal collection for $\Qq$. %
    Next, let $V_X$ be defined as $U \sm N_G[V_4]$, $\ell = |V_X|$, and
    $f$ be any bijection $X_\ell \to V_X$.
    
    For each $P \in \Qq$, we construct the trace $T_P$ by orienting $P$ in any direction, recording the occurrences of vertices from $N_G[V_4] \cup V_X$,
    and replacing each occurrence of $v \in V_X$ with $f^{-1}(v) \in X_\ell$.
    The obtained pair is $(\Tt, d)$ where $\Tt$ comprises traces $T_P$, for each $P \in \Qq$, and $d \colon X_\ell \to \{1,2,3\}$ is defined as $d(x_i) = \deg_G(f(x_i))$.
\end{definition}

The outcome of the above construction is not unique because the orientations of the paths, the terminal collection, and the bijection $f$ are arbitrary.

\subsection{Pattern definition}

In order to define a combinatorial object describing the encoding, we first need more notation to work with traces.

\begin{definition}[Trace Notation]
    For a $(G, V_4, \ell)$-trace $T = (t_1,t_2,t_3,\dots,t_r)$ we define $\edges(T)$ as the set of those unordered pairs among $t_1t_2,\, t_2t_3, \dots, t_{r-1}t_r$ that contain a symbol from $V_4$.
    Next, we define $\pairs(T)$ as the set of the remaining pairs from the above list.   
    For a collection $\cal T$ of traces, we define $\edges(\Tt)$ (resp. $\pairs(\Tt))$ as the union of $\edges(T)$ (resp. $\pairs(T))$ over $T \in \cal T$. 
    For $y \in N_G[v_4] \cup X_\ell$ we define $\deg_T(y)$ as:
    \begin{enumerate}
        \item[0:] if $y$ does not appear on $T$,
        \item[1:] if $y$ is an endpoint of $T$,
        \item[2:] otherwise.
     \end{enumerate}
\end{definition}

We are finally ready to properly define a {\em pattern}. 

\begin{definition}[Pattern]\label{def:pattern}
    A {\em $(G,V_4,\ell)$-pattern} is a pair $({\cal T}, d)$, where $\cal T$ is a collection of $(G,V_4,\ell)$-traces and $d$ is a function $X_\ell \to \{1,2,3\}$
    with the following properties.
    \begin{enumerate}
        \item The sets $\pairs(T_1), \pairs(T_2)$ are disjoint for distinct $T_1,T_2 \in \cal T$.\label{item:pattern:ends}
        \item The sets $\edges(T_1), \edges(T_2)$ are disjoint for distinct $T_1,T_2 \in \cal T$.\label{item:pattern:edges}
        \item The set $\edges(\Tt)$ equals the set of edges incident to $V_4$ in $G$.\label{item:pattern:cover}
        \item For each $x_i\in X_\ell$ we have $\sum_{T\in\Tt} \deg_T(x_i)\le d(x_i)$.\label{item:pattern:deg-x}
        \item For each $v \in N_G(V_4)$ we have $\sum_{T\in\Tt} \deg_T(v)\le \deg_G(v)$.\label{item:pattern:deg-n}
    \end{enumerate}
\end{definition}

The conditions (2, 3) ensure that a pattern encodes all the information about the edges incident to $V_4$, i.e., which path utilizes which edges.
The last two conditions reflect the fact that the number of times a vertex is ``used'' in a pattern is bounded by its degree.
They will also come in useful for estimating the total number of different patterns.
Next, we need to justify that any $(\Tt, d)$ obtained via \Cref{def:match} is indeed a valid pattern.
We also utilize \Cref{lem:terminal} to bound the number of variables needed in an encoding.

\begin{restatable}{lemma}{lemPatternEncoding}
\label{lem:pattern-encoding}
    Let $(G,V_4)$ be a nice pair, $\Qq$ be a covering family,
    and $(\Tt, d)$ be a pair constructed according  to \Cref{def:match}.
    Then $(\Tt, d)$ is a $(G,V_4,\ell)$-pattern.
    Moreover, there exists  a $(G,V_4,\ell)$-pattern encoding $\Qq$ which satisfies  $\ell \le 16\cdot \high(G)$.
\end{restatable}

\lv{%
\begin{proof}
First we show that each $T \in \Tt$ is a $(G,V_4,\ell)$-trace.
Clearly, each symbol appears at most once on $T$ because a path $P$ visits each vertex at most once.
Next, no variable $x_i$ appears next to $v \in V_4$ because the vertices consecutive to $v$ on $P$ belong to $N_G[V_4]$ while $V_X \cap N_G[V_4] = \emptyset$.

Now we check the conditions of \Cref{def:pattern}.
Since the set $U$ was chosen as a terminal collection for $\Qq$, no pair of vertices appears twice as consecutive in the collection of traces.
This ensures that the first two conditions hold.
Next, let $vu \in E(G)$ be an edge incident to $v \in V_4$.
Then $vu$ appears on some path $P \in \Qq$ and so $vu \in \edges(\Tt)$.
To see the last two conditions, observe that when a trace $T$ encodes a path $P$ then for each  $v \in N_G[V_4]$ we have $\deg_T(v) = \deg_P(v)$.
Furthermore, when $x_i \in X_\ell$ and $u =f(x_i) \in V_X$ then $\deg_T(x_i) = \deg_P(u)$ and $d(x_i) = \deg_G(u)$.

Finally, we argue that $\ell$ can be bounded by $16\cdot \high(G)$.
In \Cref{def:match} we can utilize any set $U$ that forms a terminal collection for $\Qq$.
\Cref{lem:terminal} shows that there exists such $U$ of size at most $16\cdot \high(G)$.
\end{proof}
}

\section{Decoding a Pattern}\label{sec:MAIN}

We will now explain how $\pn(G)$ can be determined by examining patterns appearing in $G$.

\begin{definition}[Odd Number of a Pattern]\label{def:OddPattern}
Let $(G,V_4)$ be a nice pair, $\ell \in \mathbb{N}$, and $(\Tt, d)$ be a $(G,V_4,\ell)$-pattern.
We say that a vertex $v \in N_G[V_4]$ (resp. variable $v \in X_\ell$) {\em loses oddity} if $\deg_G(v)$ (resp. $d(v))$ is \underline{odd} and $\sum_{T\in\Tt}\deg_T(v)$ is \underline{odd}.
Similarly, a vertex $v \in N_G[V_4]$ (resp. variable $v \in X_\ell$) {\em gains oddity} if $\deg_G(v)$ (resp. $d(v))$ is \underline{even} and $\sum_{T\in\Tt}\deg_T(v)$ is \underline{odd}.

We define the \emph{odd number} $\odd(\Tt, d)$ of a pattern $(\Tt,d)$ as the number of elements in $N_G[V_4] \cup X_\ell$ that gain oddity minus the number of elements that lose oddity.
\end{definition}

The utility of the above definition stems from the following lemma, which allows us to read $\odd(G-\Qq)$ by examining only the pattern encoding $\Qq$.

\begin{lemma}\label{lem:oddity-remainder}
Let $(G,V_4)$ be a nice pair,
$\Qq$ be a covering family in $(G,V_4)$, 
$\ell \in \mathbb{N}$, and $({\cal T}, d)$ be a $(G,V_4,\ell)$-pattern encoding $\Qq$.
Then $\odd(G - \Qq) = \odd(G) + \odd(\Tt,d)$.
\end{lemma}
\begin{proof}
The quantity $\odd(G - \Qq) - \odd(G)$ equals the number of vertices $v \in V(G)$ that are odd in $G-\Qq$ and even in $G$, minus the number of vertices that are even in $G-\Qq$ and odd in $G$.
This transition occurs exactly when  $\sum_{P\in\Qq}\deg_P(v) = \sum_{T\in\Tt}\deg_T(v)$ is odd.
Observe that when $v \in V(G)$ is not an endpoint of any path $P \in \Qq$ then its oddity is not affected, so it suffices to inspect only the vertices recorded in the encoding.
When  $\sum_{P\in\Qq}\deg_P(v)$ is odd then,
depending on whether $\deg_G(v)$ is even or odd, vertex $v$  gains or loses oddity, respectively.
Consequently, $\odd(G - \Qq) - \odd(G) = \odd(\Tt,d)$.
\end{proof}

The next two lemmas show that $\pn(G)$ can be expressed as $\min\left(\frac{\odd(G)+\odd(\Tt, d)}{2} +|\Tt|\right)$ over all patterns $({\cal T}, d)$ that are {\em feasible} in $G$.
Recall that for a path partition $\Pp$, we write $\Pp_4 \sub \Pp$ to denote the subfamily of paths visiting any vertex from $V_4$ and that $\Pp_4$ forms a covering family.
We first show that a pattern encoding $\Pp_4$ provides a lower bound on $|\Pp|$.

\begin{lemma}\label{lem:pn-greater-odd}
Let $(G,V_4)$ be a nice pair,
$\Pp$ be a path partition of $G$, $\Pp_4 \sub \Pp$ be the corresponding covering family,
$\ell \in \mathbb{N}$, and $({\cal T}, d)$ be a $(G,V_4,\ell)$-pattern encoding $\Pp_4$. 
Then \[|\Pp| \ge \frac{\odd(G)+\odd(\Tt, d)}{2} +|\Tt|.
 \]
\end{lemma}

\begin{proof}
Recall that $|\Tt|=|\Pp_4|$ by \cref{def:match}.
Let $G'\coloneqq G-\Pp_4$. 
By \Cref{obs:lower-bound} we have $\pn(G') \ge \odd(G')/2$.
Next, \Cref{lem:oddity-remainder} yields
$\mathrm{odd}(G')=\mathrm{odd}(G)+\mathrm{odd}(\Tt, d)$ and so
\[
|\Pp|-|\Tt| = |\Pp \sm \Pp_4| \ge \pn(G') \ge \frac{\odd(G')}{2} = \frac{\odd(G)+\odd(\Tt, d)}{2}. \qedhere
\]
\end{proof}

The more complicated direction is to show that the odd number of a feasible pattern can be used to attain an upper bound on $\pn(G)$.
This argument relies on \Cref{thm:subcubic-case} for subcubic graphs and requires the encoded covering family to be bull-free.
Recall from \Cref{lem:bull-free} that this holds in any optimal solution.

\begin{lemma}\label{lem:pn-less-odd}
Let $(G,V_4)$ be a nice pair,
$\Qq$ be a bull-free covering family in $(G,V_4)$, 
$\ell \in \mathbb{N}$, and $({\cal T}, d)$ be a $(G,V_4,\ell)$-pattern encoding $\Qq$.
Then %
 \[\pn(G) \le \frac{\mathrm{odd}(G)+\mathrm{odd}(\Tt, d)}{2} +|\Tt|.
 \]
\end{lemma}

\begin{proof}

We prove the lemma by constructing a path partition $\Pp_0$ for $G$ such that
\begin{align*}
|\Pp_0| = \frac{\mathrm{odd}(G)+\mathrm{odd}(\Tt, d)}{2} +|\Tt|.
\end{align*}

Let $G':=G-\Qq$. As $G'$ is subcubic, we can apply \Cref{thm:subcubic-case} for every subgraph of $G'$. The main idea is to extend some paths in $\Qq$ to obtain another bull-free covering family $\Qq'$ of the same size in a way that the subgraph $G'':=G-\Qq'$ satisfies $\mathrm{odd}(G'')=\mathrm{odd}(G')$ and it has a path partition $\Qq$ of size $\mathrm{odd}(G')/2$. 
This occurs when $G''$ has no pan cycle and no connected component which is either a cycle or a subdivided diamond.
To construct $\Qq'$, we gradually extend paths in $\Qq$ in three steps. Note that $\deg_{G}(u)\leq 3$ for all non-isolated vertices $u\in V(G')$ since $\Qq$ is a covering family. This, in particular, implies that for every vertex $u\in V(G')$ satisfying $\deg_{G'}(u)=2\neq \deg_{G}(u)$, we have $\deg_{G}(u)=3$, and there exists a unique path in $\Qq$ containing $u$, where $u$ is in fact an endpoint of this path.

\begin{figure}[!h]
    \centering
\begin{tikzpicture}[scale=0.4]
\node (v4) at (0,0) [circle, draw] {};
\node (v5) at (3,0) [circle, draw] {};
\node (v1) at (4.5,2.5) [circle, draw, inner sep=3pt] {$w$};
\node (v2) at (3,5) [circle, draw] {};
\node (v3) at (0,5) [circle, draw, inner sep=3pt] {$v$};
\node (v') at (-7,2) [circle, draw, inner sep=3pt] {$v'$};
\node (v6) at (7,2.5) [circle, draw] {};

\node at (-4,4.4) {\color{red} \Large $P$};

\node at (2.4,4.1) {\color{red} \large $C_0$};

\draw [line width=2pt,color=red] (v1)--(v2);
\draw [line width=2pt] (v1)--(v5);
\draw[line width=2pt] (v1)--(v6);
\draw [line width=2pt] (v4)--(v5);
\draw [line width=2pt] (v3)--(v4);
\draw[line width=2pt,color=red] (v2)--(v3);

   \draw[line width=2pt,color=red]
        (v3)
        .. controls (-3,1) and (-5,6) ..
        (v');

\node (1v4) at (15,0) [circle, draw, inner sep=3pt] {$v'$};
\node (1v5) at (18,0) [circle, draw] {};
\node (1v1) at (19.5,2.5) [circle, draw, inner sep=3pt] {$w$};
\node (1v2) at (18,5) [circle, draw] {};
\node (1v3) at (15,5) [circle, draw, inner sep=3pt] {$v$};
\node (1v6) at (22,2.5) [circle, draw] {};

\draw [line width=2pt,color=red] (1v1)--(1v2);
\draw [line width=2pt] (1v1)--(1v5);
\draw[line width=2pt] (1v1)--(1v6);
\draw [line width=2pt] (1v4)--(1v5);
\draw [line width=2pt] (1v3)--(1v4);
\draw[line width=2pt,color=red] (1v2)--(1v3);

   \draw[line width=2pt,color=red]
        (1v3)
        .. controls (10,3) and (12,1) ..
        (1v4);

\node at (12,4.4) {\color{red} \Large $P$};

\node at (17.4,4.1) {\color{red} \large $C_0$};

\node at (0,-2) {\color{black} when $v'\notin V(C)$};

\node at (17,-2) {\color{black} when $v'\in V(C)$};

\end{tikzpicture}

\caption{We can extend $P\in \Qq$ to discard the pan cycle $C$ in $G'$.}
\label{fig:KILL-PAN-CYCLE}
\end{figure}
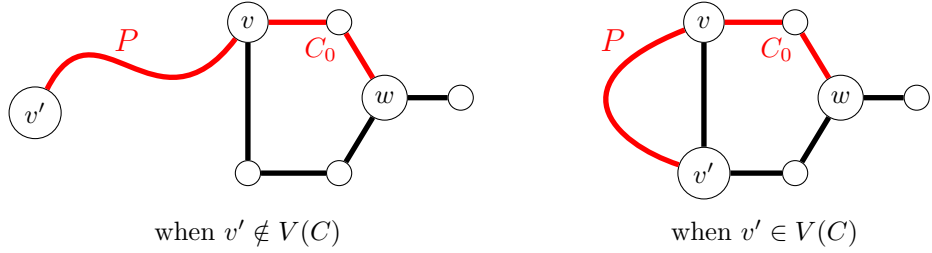

First, assume that $G'$ has a pan cycle $C$, and let $w\in V(C)$ be the unique vertex satisfying $\deg_{G'}(w)=3$. It follows that $\deg_{G}(w)=3$ using $\deg_{G'}(w)\leq \deg_{G}(w)\leq 3$. 
As $G$ is nice, it contains no pan cycle and so there exists $v\in V(C)$ satisfying $\deg_{G}(v)\neq \deg_{G'}(v)=2$. This implies that $v\neq w$, and that there exists an $(v,v')$-path $P\in \Qq$ for some $v'\in V(G)$.
Note that $V(C)\cap (V(P)\setminus \{v,v'\})=\emptyset$ because every vertex from $V(G) \sm V_4$ that is internal in $P$ has degree at most one in $G'$.
If $v'\in V(C)$, let $C_0$ be the $(w,v)$-path on $C$ not containing $v'$. Otherwise, let $C_0$ be any of the two $(w,v)$-paths on $C$. By extending $P$ to $v'PvC_0w$, depicted in Figure~\ref{fig:KILL-PAN-CYCLE}, we can decrease the number of pan cycles in $G'$ by ensuring that $\odd(G')$ remains the same because we only reverse the parities of $\deg_{G'}(v)$ and $\deg_{G'}(w)$. Moreover, it is clear that $v'$ and $w$ do not form a bull-pair. By applying the same procedure successively for every pan cycle in $G'$, we end up with a bull-free covering family $\Qq_1$ satisfying $|\Qq_1|=|\Qq|$ such that the subgraph $G_1:=G-\Qq_1$ satisfies $\pan(G_1)=0$ and $\odd(G_1)=\odd(G')$.

Second, assume that $G_1$ has a connected component that is a cycle $C$. Since $G$ is connected and not subcubic, there are edges between $V(C)$ and $V(G)\setminus V(C)$, which implies that $W:=\{w\in V(C):\deg_{G}(w)\neq 2\}$ is non-empty. Moreover, using $\pan(G)=0$, we see that $|W|\neq 1$. For every $w\in W$, we write $P_w$ for the unique path in $\Qq_1$ containing $w$. In fact, $w$ is an endpoint of $P_w$. 
As before, we observe that no internal vertex of $P_w$ can lie on $C$.

If $|W|\geq 3$, we can choose $x,y\in W$ for which $P_x\neq P_y$. If $|W|=2$, letting $W=\{x,y\}$, we claim that $P_x\neq P_y$. Indeed, when $|W|=2$, $C$ is a bull cycle in $G$. As $G$ is nice, it follows that $C$ is in fact a triangle. As $\Qq_1$ is bull-free and $\{x,y\}$ form a bull-pair, it follows that $P_x\neq P_y$. Let $x'$ and $y'$ be the other endpoints of $P_x$ and $P_y$, respectively. We will extend $P_x$ and $P_y$ to cover $C$ depending on whether $x'$ and $y'$ lie on $C$.

If at most one of $x'$ and $y'$ belongs to $C$, without loss of generality, assume that $y'\notin V(C)$. If $x'\in V(C)$, let $C_x$ be the $(x,y)$-path on $C$ not containing $x'$. If $x'\notin V(C)$, let $C_x$ be any $(x,y)$-path on $C$. Let $C_y$ be the $(x,y)$-path on $C$ other than $C_x$. Then, we can extend $P_x$ and $P_y$ to $x'P_xxC_xy$ and $y'P_yyC_yx$, respectively, depicted in Figure~\ref{fig:KILL-CYCLE-COMPONENT-I}.
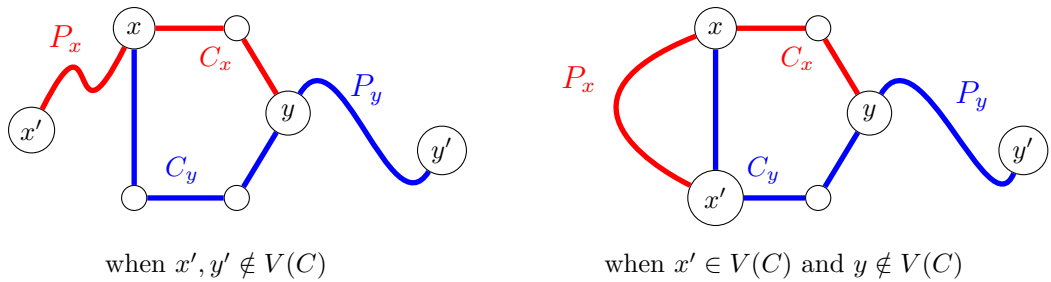
\begin{figure}[h]
    \centering
\begin{tikzpicture}[scale=0.45]
\node (v4) at (0,0) [circle, draw] {};
\node (v5) at (3,0) [circle, draw] {};
\node (v1) at (4.5,2.5) [circle, draw, inner sep=3pt] {$y$};
\node (v2) at (3,5) [circle, draw] {};
\node (v3) at (0,5) [circle, draw, inner sep=3pt] {$x$};
\node (v') at (-3,2) [circle, draw, inner sep=2.1pt] {$x'$};
\node (v6) at (9,1.4) [circle, draw, inner sep=2.1pt] {$y'$};

\node at (-2,4.7) {\color{red} \Large $P_x$};

\node at (2.4,4.1) {\color{red} \large $C_x$};

\node at (6.8,3.2) {\color{blue} \Large $P_y$};

\node at (1.4,0.8) {\color{blue} \large $C_y$};

\draw [line width=2pt,color=red] (v1)--(v2);
\draw [line width=2pt,color=blue] (v1)--(v5);
\draw [line width=2pt,color=blue] (v4)--(v5);
\draw [line width=2pt,color=blue] (v3)--(v4);
\draw[line width=2pt,color=red] (v2)--(v3);

   \draw[line width=2pt,color=red]
        (v3)
        .. controls (-2,1) and (-1,6) ..
        (v');

   \draw[line width=2pt,color=blue]
        (v1)
        .. controls (6,5) and (7.5,-1) ..
        (v6);

\node (11v4) at (17,0) [circle, draw] {$x'$};
\node (11v5) at (20,0) [circle, draw] {};
\node (11v1) at (21.5,2.5) [circle, draw, inner sep=3pt] {$y$};
\node (11v2) at (20,5) [circle, draw] {};
\node (11v3) at (17,5) [circle, draw, inner sep=3pt] {$x$};
\node (11v6) at (26,1.4) [circle, draw, inner sep=2.1pt] {$y'$};

\node at (13,3.5) {\color{red} \Large $P_x$};

\node at (19.4,4.1) {\color{red} \large $C_x$};

\node at (24.5,3) {\color{blue} \Large $P_y$};

\node at (18.4,0.8) {\color{blue} \large $C_y$};

\draw [line width=2pt,color=red] (11v1)--(11v2);
\draw [line width=2pt,color=blue] (11v1)--(11v5);
\draw [line width=2pt,color=blue] (11v4)--(11v5);
\draw [line width=2pt,color=blue] (11v3)--(11v4);
\draw[line width=2pt,color=red] (11v2)--(11v3);

   \draw[line width=2pt,color=red]
        (11v3)
        .. controls (12,3) and (15,1) ..
        (11v4);

   \draw[line width=2pt,color=blue]
        (11v1)
        .. controls (23,5) and (25,-1) ..
        (11v6);

\node at (2.4,-2) {\color{black} when $x',y'\notin V(C)$ };

\node at (19,-2) {\color{black} when $x'\in V(C)$ and $y\notin V(C)$};

\end{tikzpicture}

\caption{When at most one of $x'$ and $y'$ belongs to $C$, then two $(x,y)$-paths of $C$ can be appended to $P_x$ and $P_y$.}
\label{fig:KILL-CYCLE-COMPONENT-I}
\end{figure}

If $x',y'\in V(C)$, without loss of generality, we can assume that $\{x,x',y,y'\}$ lie on the cycle $C$ with the order $(x,x',y,y')$ or $(x,y,x',y')$ clockwise. For the first case, let $C_x$ (resp.~$C_y$) be the $(x',y')$-path on $C$ containing $x$. We can extend $P_x$ and $P_y$ to $xP_xx'C_yy'$ and $yP_yy'C_xx'$, respectively, depicted in Figure~\ref{fig:KILL-CYCLE-COMPONENT-II} (\mic{left}). For the second case, we can partition $C$ into an $(x,y)$-path $C_1$, a $(y,x')$-path $C_2$, an $(x',y')$-path $C_3$, and a $(y',x)$-path $C_4$ such that $E(C)=E(C_1)\cup E(C_2)\cup E(C_3)\cup E(C_4)$. We can extend $P_x$ and $P_y$ to $yC_{1}xP_xx'C_{3}y'$ and $x'C_{2}yP_yy'C_{4}x$, respectively, as shown in Figure~\ref{fig:KILL-CYCLE-COMPONENT-II} (\mic{right}). 
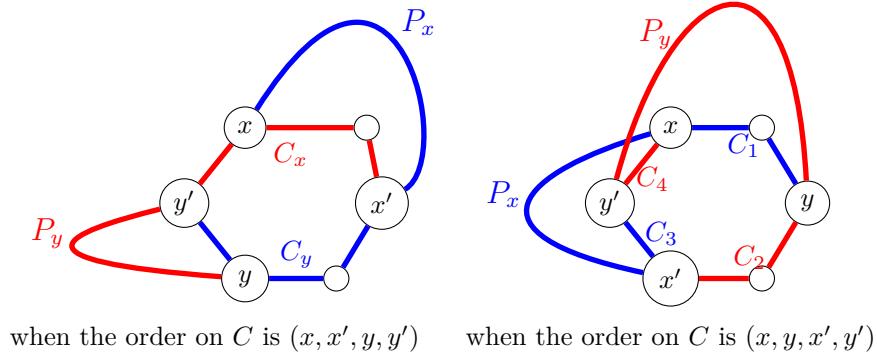
\begin{figure}[!h]
    \centering
\begin{tikzpicture}[scale=0.4]
\node (y') at (0,2.5) [circle, draw, inner sep=2.1pt] {$y'$};
\node (y) at (2,0) [circle, draw] {$y$};
\node (v1) at (5,0) [circle, draw] {};
\node (x') at (6.5,2.5) [circle, draw, inner sep=3pt] {$x'$};
\node (v2) at (6,5) [circle, draw] {};
\node (x) at (2,5) [circle, draw, inner sep=3pt] {$x$};

\node at (7.7,8.5) {\color{blue} \Large $P_x$};

\node at (3.5,4.1) {\color{red} \large $C_x$};

\node at (-4.5,1.5) {\color{red} \Large $P_y$};

\node at (3.7,0.8) {\color{blue} \large $C_y$};

\draw [line width=2pt,color=red] (x)--(v2);
\draw [line width=2pt,color=red] (v2)--(x');
\draw [line width=2pt,color=blue] (x')--(v1);
\draw [line width=2pt,color=blue] (v1)--(y);
\draw [line width=2pt,color=blue] (y)--(y');
\draw [line width=2pt,color=red] (y')--(x);

   \draw[line width=2pt,color=blue]
        (x)
        .. controls (7,13) and (9,4) ..
        (x');

   \draw[line width=2pt,color=red]
        (y)
        .. controls (-6,0.6) and (-4,1.5) ..
        (y');

\node (1y') at (14,2.5) [circle, draw, inner sep=2.1pt] {$y'$};
\node (1x') at (16,0) [circle, draw] {$x'$};
\node (1v1) at (19,0) [circle, draw] {};
\node (1y) at (20.5,2.5) [circle, draw, inner sep=3pt] {$y$};
\node (1v2) at (19,5) [circle, draw] {};
\node (1x) at (16,5) [circle, draw, inner sep=3pt] {$x$};

\draw [line width=2pt,color=blue] (1x)--(1v2);
\draw [line width=2pt,color=blue] (1v2)--(1y);
\draw [line width=2pt,color=red] (1y)--(1v1);
\draw [line width=2pt,color=red] (1v1)--(1x');
\draw [line width=2pt,color=blue] (1x')--(1y');
\draw[line width=2pt,color=red] (1y')--(1x);

\node at (10.5,2.8) {\color{blue} \Large $P_x$};

\node at (15.5,8.1) {\color{red} \Large $P_y$};

\node at (18.4,4.4) {\color{blue} \large $C_1$};
\node at (18.6,0.7) {\color{red} \large $C_2$};

\node at (15.7,1.5) {\color{blue} \large $C_3$};
\node at (15.4,3.3) {\color{red} \large $C_4$};

   \draw[line width=2pt,color=blue]
        (1x)
        .. controls (11,3.5) and (9,1.5) ..
        (1x');

   \draw[line width=2pt,color=red]
        (1y)
        .. controls (20,12) and (16,10) ..
        (1y');

\node at (1,-2) {\color{black} when the order on $C$ is $(x,x',y,y')$};

\node at (16,-2) {\color{black} when the order on $C$ is $(x,y,x',y')$};

\end{tikzpicture}

\caption{If $x',y'\in V(C)$, we can partition $C$ into either two or four parts, depending on the order of $x,y,x',y'$ on $C$, such that every part can be appended to $P_x$ or $P_y$.}
\label{fig:KILL-CYCLE-COMPONENT-II}
\end{figure}

In all cases, we ensure that $\odd(G_1)$ remains the same as we only change degrees from 2 to 0, and we do not create a bull-pair. By applying the same procedure for every connected component that is a cycle, we end up with a bull-free covering family $\Qq_2$ satisfying $|\Qq_2|=|\Qq_1|=|\Qq|$ such that the subgraph $G_2:=G-\Qq_2$ satisfies $\pan(G_2)=0$, $\odd(G_2)=\odd(G_1)=\odd(G')$, and $G_2$ has no connected component that is a cycle.

Third, let $G_2$ has a connected component $D$ that is a subdivision of a diamond. Let $x,y\in V(D)$ be the vertices with $\deg_{G_2}(x)=3=\deg_{G_2}(y)$. Note that $\deg_{G}(x)=3=\deg_{G}(y)$. Since $G$ is not subcubic and connected, there are edges between $V(D)$ and $V(G)\setminus V(D)$, which implies that there exists $z\in V(D)$ satisfying $\deg_{G}(z)\neq \deg_{G'}(z)=2$. Note that we can partition $D$ into an $(x,z)$-path $D_y$ and a $(y,z)$-path $D_x$ with $E(D)=E(D_x)\cup E(D_y)$, depicted in Figure~\ref{fig:KILL-DIAMOND-COMPONENT}. It is clear that every $u\in V(D)\setminus \{x,y,z\}$ appears exactly one of $D_x$ and $D_y$. Let $P_z$ be the unique path in $\Qq_2$ containing $z$, and let $z'$ be the other endpoint of $P_z$. As $z'\notin V(D_x)\cap V(D_y)$, without loss of generality, we can assume that $z'\notin V(D_y)$. Then, we can extend $P_z$ to $z'P_zzD_yx$ regardless of whether $z'\in V(D_x)$ or not. Note that we ensure that $\odd(G_2)$ remains the same because we 
only reverse the parities of $\deg_{G_2}(z)$ and $\deg_{G_2}(x)$. 
Effectively, the component $C$ of $G'$ with $\odd(C)=2$ is replaced by the path $D_x$, also with $\odd(D_x)=2$.

\begin{figure}[!h]
    \centering
\begin{tikzpicture}[scale=0.34]
\node (y) at (10,0) [circle, draw, inner sep=2.1pt] {$y$};
\node (x) at (0,0) [circle, draw, inner sep=2.1pt] {$x$};
\node (z) at (4,4.5) [circle, draw, inner sep=2.1pt] {$z$};
\node (z') at (12,4.5) [circle, draw, inner sep=1.7pt] {$z'$};

   \draw[line width=2pt,color=blue]
        (x)
        .. controls (6.5,-1) and (7,1.7) ..
        (y);

   \draw[line width=2pt,color=red]
        (x)
        .. controls (-2,4.5) and (3,2) ..
        (z);

   \draw[line width=2pt,color=blue]
        (z)
        .. controls (6,2) and (8,4) ..
        (y);

   \draw[line width=2pt,color=red]
        (x)
        .. controls (2,-4) and (3.5,-5) ..
        (y);

    \draw[line width=2pt,color=blue]
        (z)
        .. controls (7,7) and (10,1.5) ..
        (z');

\node at (6.8,2) {\color{blue} \large $D_y$};
\node at (2.9,-2.6) {\color{red} \large $D_x$};
\node at (7.8,5.6) {\color{black} \large $P_z$};

\node (1y) at (30,0) [circle, draw, inner sep=2.1pt] {$y$};
\node (1x) at (20,0) [circle, draw, inner sep=2.1pt] {$x$};
\node (1z) at (24,4.5) [circle, draw, inner sep=2.1pt] {$z$};
\node (1z') at (25,-4) [circle, draw, inner sep=1.7pt] {$z'$};

  \draw[line width=2pt,color=blue]
        (1x)
        .. controls (26.5,-1) and (27,1.7) ..
        (1y);

   \draw[line width=2pt,color=red]
        (1x)
        .. controls (18,4.5) and (23,2) ..
        (1z);

   \draw[line width=2pt,color=blue]
        (1z)
        .. controls (26,2) and (28,4) ..
        (1y);

   \draw[line width=2pt,color=red]
        (1x)
        .. controls (21.5,-4) and (23.3,-5) ..
        (1z');

    \draw[line width=2pt,color=red]
        (1z')
        .. controls (26.5,-2.5) and (28,-1) ..
        (1y);

\draw[line width=2pt,color=blue]
        (1z)
        .. controls (37,3) and (34,-2) ..
        (1z');

\node at (26.8,2) {\color{blue} \large $D_y$};
\node at (22.3,-2.6) {\color{red} \large $D_x$};
\node at (34,2.3) {\color{black} \large $P_z$};

\node at (5,-6.6) {\color{black} when $z'\notin V(D_x)$};

\node at (26,-6.6) {\color{black} when $z'\in V(D_x)$};

\end{tikzpicture}

\caption{The partition of a diamond $D$ into two edge-disjoint paths $D_x$ and $D_y$. The path $P_z$ can intersect with $D$ only at $z$ and $z'$.}
\label{fig:KILL-DIAMOND-COMPONENT}
\end{figure}
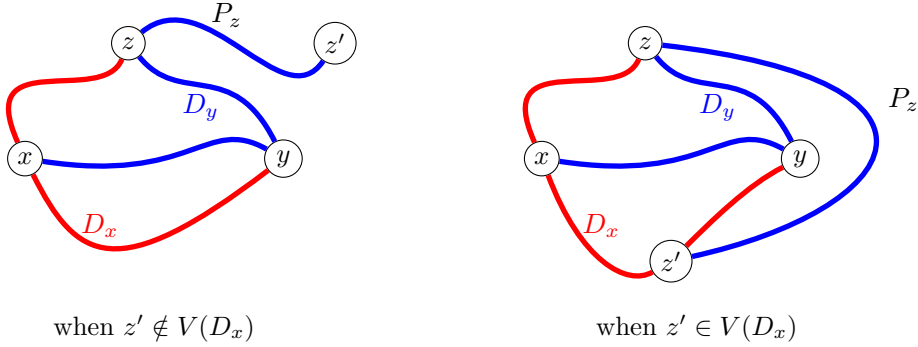

By this way we decrease the number of connected components of $G_2$ that are subdivisions of a diamond. It is also clear that we do not create any bull-pair. By applying this procedure exhaustively, we end up with a bull-free covering family $\Qq'$ satisfying $|\Qq'|=|\Qq_2|=|\Qq_1|=|\Qq|$ such that the subgraph $G'':=G-\Qq'$ satisfies $\pan(G')=0$, $\odd(G'')=\odd(G_2)=\odd(G_1)=\odd(G')$, and no connected component of $G''$ is a cycle or a subdivision of a diamond. 

We infer from \Cref{thm:subcubic-case} that $\pn(G'')=\odd(G'')/2=\odd(G')/2$. 
Hence, there exists a path partition $\Pp''$ for $G''$ of size $\odd(G')/2$. Set $\Pp_0=\Pp''\,\cup \Qq'$. Recall that $\odd(G')=\odd(G)+\odd(\Tt,d)$ by \Cref{lem:oddity-remainder}. %
As $|\Tt| = |\Qq|$, we arrive at
\begin{align*}
\pn(G)\le |\Pp_0|=\frac{\odd(G'')}{2}+|\Qq'|=\frac{\odd(G')}{2}+|\Qq|=\frac{\odd(G)+\odd(\Tt, d)}{2} +|\Tt|.
\end{align*}
This concludes the proof of \Cref{lem:pn-less-odd}.
\end{proof}

\section{Parameterized Algorithm}
\label{sec:ALGO}

Recall that two paths are internally vertex-disjoint if no vertex of one path appears as an internal vertex on the other path.

In the algorithm we will need to check if a given pattern $({\cal T}, d)$ is {\em feasible}, i.e., if there exists a bull-free covering family encoded by $({\cal T}, d)$.
First, we formulate feasibility via conditions that can be expressed using a logical formula over the subcubic graph $G-V_4$.
In order to incorporate the logical meta-theorem from~\cite{Logic} we need to translate edge-disjointedness to internal vertex-disjointedness.
The first direction follows from subcubicness while the opposite implication requires at least one endpoint of each two paths to be distinct.
To ensure this, we exploit the concept of terminal collection from \cref{def:terminal}.

\begin{restatable}{lemma}{lemPatternConditions}
\label{lem:patern-conditions}
    Let $(G,V_4)$ be a nice pair, $\ell \in \mathbb{N}$, and $({\cal T}, d)$ be a $(G,V_4,\ell)$-pattern.
    Then there exists a bull-free covering family in $(G,V_4)$ encoded by  $({\cal T}, d)$
    if and only if there exists
    a mapping $f \colon X_\ell \cup N_G(V_4) \to V(G)$
    that satisfies the following:
    \begin{enumerate}
        \item $f$ is an identify on $N_G(V_4)$,
        \item $f$ is injective on $X_\ell$ and $f(X_\ell) \sub V(G) \setminus N_G[V_4]$,
        \label{item:patern:function}
        \item $\deg_G(f(x_i)) = d(x_i)$ for each $x_i \in X_\ell$,
        \label{item:patern:degree}
        \item there exists a collection of internally vertex-disjoint paths in  $G-V_4$ connecting the pairs in the collection $\left\{(f(y_1), f(y_2)) \colon y_1y_2 \in \pairs(\Tt)\right\}$,
        \label{item:patern:paths}
        \item for each $T \in \Tt$ with endpoints $y_1,y_2$, the pair $(f(y_1), f(y_2))$ is not a bull-pair in $G$.
        \label{item:patern:bull}
    \end{enumerate}
\end{restatable}

\lv{%

\begin{proof}
First we argue that when $(\Tt,d)$ encodes a covering family $\Qq$ then the given conditions are satisfied.
The only condition that is not obvious is \eqref{item:patern:paths}.
By the definition of the set $\pairs(\Tt)$, for each $y_1y_2 \in \pairs(\Tt)$ there exists a path $P \in \Qq$ containing a subpath $Q$ from $f(y_1)$ to $f(y_2)$.
We need to prove that these subpaths are pairwise internally vertex-disjoint and contained in $G-V_4$.
To see the latter fact, note that, by definition, $f(y_1),f(y_2) \not\in V_4$ and $Q$ cannot visit any $v\in V_4$ because then $v$ would be recorded in the trace encoding $P$, implying that $y_1,y_2$ cannot be consecutive in this trace.

Next, we prove internal vertex-disjointedness.
First consider the case where two subpaths $Q_1,Q_2$ contained in a single path  $P \in \Qq$.
By the definition of a trace, these subpaths do not overlap on $P$.
Since $P$ visits each vertex at most once, the only common vertex on $Q_1,Q_2$ may be their common endpoint.
Now suppose that $Q_1 \sub P_1$ and $Q_2 \sub P_2$ for $P_1,P_2 \in \Qq$ and assume w.l.o.g. that there exists $v \in V(Q_1) \cap V(Q_2)$ and $v$ is internal on $Q_1$.
If $v$ is also internal on $Q_2$ then $\deg_G(v) \ge 4$ because the paths $P_1,P_2$ are edge-disjoint.
But then $v \in V_4$ and each occurrence of $v$ is recorded in each trace $T \in \Tt$, implying that $y_1,y_2$ cannot be consecutive in the traces encoding $P_1$ or $P_2$; contradiction.
It follows that $v$ is an endpoint of $Q_2$.
But then $v \in V_X$ and again every occurrence of $v$ is recorded in each trace $T \in \Tt$, leading to the same contradiction as above.
This concludes the justification of condition  \eqref{item:patern:paths} and the first part of the proof.

Now we argue that the given conditions are sufficient for an existence of a bull-free covering family encoded by $(\Tt,d)$.
Let $V_X = f(X_\ell)$.
By condition \eqref{item:patern:function} we have $|V_X| = \ell$ and $V_X \cap N_G[V_4] = \emptyset$.
The function $d$ is defined in \Cref{def:match} exactly as in condition \eqref{item:patern:degree}.

For each $y_1y_2$ in $\pairs(\Tt)$ there exists a path $Q(y_1,y_2)$ from $f(y_1)$ to $f(y_2)$ so that these paths are internally vertex-disjoint.
For a trace $T \in \Tt$ let $P_T$ be the concatenation of {\em subpaths} $Q(y_1,y_2)$ for $y_1y_2 \in \pairs(T)$ and edges from $\edges(T)$, in the order specified in $T$.
Then $P_T$ is a walk and to argue that $P_T$ is a path we need to prove that each vertex $v \in V(G)$ appears on $P_T$ at most once.
By the definition of a trace, each symbol occurs on $T$ at most once, so there may be at most two considered edges or subpaths having $v$ as an endpoint, and then they overlap at $v$.
Next, the subpaths $Q(y_1,y_2)$ are pairwise internally vertex-disjoint, that is, if $v$ occurs on any of them then it cannot be an internal vertex on another one.
Hence $P_T$ is indeed a path.

Next, we show that the paths $P_{T_1}$, $P_{T_2}$ are edge-disjoint for distinct $T_1,T_2 \in \Tt$.
Consider an edge $uv$ that appears on $P_{T_1}$.
If it is incident to $V_4$ then $uv \in \edges(T_1)$ and we use the fact that sets $\edges(T_1)$ and $\edges(T_2)$ are disjoint so $uv$ cannot appear on $P_{T_2}$.
Otherwise, $uv$ appears on some subpath $Q(y_1,y_2)$ for $y_1y_2$ in $\pairs(T_1)$.
If any of $u,v$ is internal on $Q(y_1,y_2)$ then, by internal vertex-disjointedness, it cannot appear in any other path in the collection.
In the remaining case, $u,v$ are the endpoints of $Q(y_1,y_2)$.
But the sets $\pairs(T_1)$ and $\pairs(T_2)$ are disjoint so
$u,v$ cannot be the endpoints of any subpath of  $P_{T_2}$.
Also, none of $u,v$ can then also appear as an internal vertex in any subpath of  $P_{T_2}$.
Hence, the paths $P_{T_1}$, $P_{T_2}$ are indeed edge-disjoint.

Finally, condition \eqref{item:patern:bull} directly corresponds to the definition of being bull-free.
Each edge incident to $V_4$ is included in $\edges(\Tt)$ (as demanded in \Cref{def:pattern}) and so it is used by one of the constructed paths.
Therefore, applying \Cref{def:match} with $\Qq$ being the constructed path family and  $V_X = f(X_\ell)$ yields the pattern $(\Tt, d)$.
The lemma follows.
\end{proof}
}

\subsection{Incorporating logic}

For a graph $G$, integer $k$, and $s_1,t_1,\dots, s_k,t_k \in V(G)$, we say that $\mathsf{dp}_k[(s_1,t_1),\dots,(s_k,t_k)]$ holds in $G$ if in $G$ there are internally vertex-disjoint paths between $s_i$ and $t_i$, for $i\in [1,k]$.

We consider first-order formulas extended by the disjoint-paths predicates over graph structures.
There are three types of an atomic formula: equality $x=y$, edge relation $E(x,y)$, and disjoint-paths relation $\mathsf{dp}_k[(x_1,y_1),\dots,(x_k,y_k)]$ for some $k \in \mathbb{N}$.
An \fodp formula is built from atomic formulas using the Boolean operators $\neg, \land, \lor$, as well as existential and universal quantification $\exists x$, $\forall x$.

A variable $x$ not in the scope of any quantifier is called a {\em free variable}.
We write $\phi(\bar x)$ to indicate that $\bar x = (x_1,\dots,x_\ell)$ is the tuple of free variables in $\phi$.
For a tuple $\bar v = (v_1,\dots,v_\ell)$ from $V(G)$ we write $G \vDash \phi(\bar v)$ if $\phi$ holds in $G$ after substitution $x_i \xleftarrow{} v_i$ for $i \in [1,\ell]$.

We summon the main theorem from \cite{Logic} in a simplified form.
The original theorem works for more general graphs classes, excluding a fixed graph $H$ as a topological minor.
Note that every subcubic graph excludes $K_{1,4}$ as a topological minor, so the statement below is indeed a corollary of the theorem from \cite{Logic}.

\begin{theorem}[\cite{Logic}]\label{thm:logic}
     There is an algorithm that, given a subcubic graph $G$, an \fodp formula $\phi(\bar x)$, and $\bar v \in V(G)^{|\bar x|}$, decides whether $G \vDash \phi(\bar v)$ in time $f(\phi)\cdot |V(G)|^3$, where $f$ is a computable function.
\end{theorem}

This theorem, together with \Cref{lem:patern-conditions}, allows us to efficiently test if a given pattern encodes some bull-free covering family.

\begin{lemma}\label{lem:pattern-testing}
    Let $(G,V_4)$ be a nice pair, $k = \high(G)$, $\ell \in \mathbb{N}$, %
    and $({\cal T}, d)$ be a $(G,V_4,\ell)$-pattern.
    Then we can check whether $({\cal T}, d)$ encodes some bull-free covering family %
    in time $g(k,\ell)\cdot |V(G)|^3$, where $g$ is a computable function.

\end{lemma}
\begin{proof}
    We shall reduce the problem to \fodp model checking over the subcubic graph $H = G-V_4$.
    Let us order $N_G(V_4)$ arbitrarily as $\bar v = (v_1,v_2,\dots, v_r)$.
    Note that $r \le k$. 
    We introduce free variables $\bar y = (y_1,y_2,\dots, y_r)$ so that $y_i$ will be assigned to $v_i$.

    We write formula $\phi(\bar y)$ starting from existential quantifiers $\exists x_1 \exists x_2 \dots \exists x_\ell$ which correspond to an assignment $X_\ell \to V(H)$.
    Together with the free variables, these encode a function $f \colon X_\ell \cup N_G(V_4) \to V(H)$ which is an~identity on $N_G(V_4)$.
    Inside the quantifiers we need to check if $f$ satisfies the conditions listed in \Cref{lem:patern-conditions}.
    
    First, it is easy to ensure that $f|_{X_\ell}$ is injective by a conjunction of $\ell \choose 2$ inequalities $x_i \ne x_j$.
    Similarly, we enforce that $f(X_\ell) \cap N_G(V_4) = \emptyset$ by writing $\ell\cdot r$ inequalities $x_i \ne y_j$.

    The next condition states that $\deg_G(f(x_i)) = d(x_i)$ for each $i \in [\ell]$.
    Note that $f(x_i) \not\in N_G(v_4)$ implies that $\deg_G(f(x_i)) = \deg_H(f(x_i))$.
    Then for each $c \in \{1,2,3\}$ it is easy to write an \fo formula checking that the degree of $x_i$ in $H$ equals $c$.

    Condition \eqref{item:patern:paths} states that the pairs of vertices in $\pairs(\Tt)$ are connected by internally vertex-disjoint paths in $H$.
    This directly corresponds to a predicate $\mathsf{dp}_h[(s_1,t_1),\dots,(s_h,t_h)]$
    where the variables $s_i,t_i$ correspond to some $x_j$ or $y_j$ depending on whether it is a symbol from $X_\ell$ or $N_G(V_4)$.
    Note that due to the degree bound each symbol can appear at most 3 times in $\pairs(\Tt)$ hence $h \le 3\cdot (k+\ell)$.

    Next, we need to ensure that for each trace $T \in \Tt$ the vertices assigned to the endpoints of $T$ do not form a bull-pair in $G$.
    Recall that $(u,v)$ is a bull-pair in $G$ when $\deg_G(u) = \deg_G(v) = 3$, $uv \in E(G)$, and they have a common neighbor $w$ of degree 2. 
    We show that this can be expressed by an \fo formula over $H$.
    We write this formula only for those traces whose endpoints have degree 3 in $G$ -- this information is known upfront because the assignment of each $y_i$ is fixed and the degree of $x_i$ is encoded using the function~$d$.
    Suppose this is the case.
    Then $u,v,w \in V(H)$.
    Moreover, $u,v$ are the only neighbors of $w$ in $G$.
    Since $u,v \not\in V_4$ we infer that $\deg_H(w) = \deg_G(w) = 2$.
    So $(u,v)$ is a bull-pair in $G$ if and only if $uv \in E(H)$ and they have a common neighbor $w$ of degree 2 in $H$.
    This again can be easily expressed in \fo over the graph $H$.
    
    We have constructed an \fodp formula $\phi(\bar y)$ so that $H \vDash \phi(\bar v)$ holds if and only if there exists a bull-free covering family in $(G,V_4)$ encoded by $({\cal T}, d)$.
    The length of this formula is $\Oh((k+\ell)^2)$.
    Finally, by \Cref{thm:logic} we can test if $H \vDash \phi(\bar v)$ in time $g(k,\ell)\cdot |V(G)|^3$.
\end{proof}

\subsection{Wrapping-up}

The final ingredient towards  \Cref{thm:main} is a bound on the number of different patterns to consider. 

\begin{restatable}{lemma}{lemPatternCount}
\label{lem:pattern-count}
    Let $(G,V_4)$ be a nice pair, $k = \high(G)$, and $\ell \in \mathbb{N}$.
    Then the number of $(G,V_4,\ell)$-patterns is $\exp(\Oh((k+\ell)\log(k+\ell)))$ and they can be enumerated in such running time.
\end{restatable}

\lv{%

\begin{proof}
    We inspect \Cref{def:pattern} to estimate the number of occurrences of each symbol from $\Sigma = N_G[V_4] \cup X_\ell$ is a pattern.
    Due to conditions \eqref{item:pattern:deg-x} and \eqref{item:pattern:deg-n} and the fact that symbols do not repeat in a single trace, each symbol from $N_G(V_4) \cup X_\ell$ can appear at most 3 times in a pattern.
    Condition \eqref{item:pattern:edges} implies that the number of occurrences of each $v \in V_4$ is bounded by $\deg_G(v)$.
    Therefore, we can bound the total length of all traces by
     $3\cdot (N_G(V_4) + \ell) + \sum_{v \in V_4} \deg_G(v) \le 4 k + 3 \ell$.

     Consequently, every $(G,V_4,\ell)$-pattern can be obtained by choosing a word of length at most $4 k + 3 \ell$ over the alphabet $\Sigma$, splitting it into subwords, and choosing function $d \colon X_\ell \to \{1,2,3\}$.
     Since $|\Sigma| \le 2k + \ell$, we can estimate the number of possibilities as $\exp(\Oh((k+\ell)\log(k+\ell)))$.
     The analysis above translates directly to an algorithm enumerating all the $(G,V_4,\ell)$-patterns.
\end{proof}
}

\begin{proof}[Proof of \cref{thm:main}]
 Each connected component of $G$ can be treated independently.
 Moreover,  \Cref{lem:niceGraphs} allows us to assume that we work with a nice pair $(G,V_4)$.
 Let $k = \high(G)$; it is bounded by $8\cdot \sen(G)$ due to
 \Cref{lem:high-sen}.
Hence, it suffices
 to give an FPT algorithm parameterized by $k$. %
 
 The algorithm enumerates all $(G,V_4,\ell)$-patterns for $\ell \le 16\cdot k$ using \Cref{lem:pattern-count} in time $\exp(\Oh(k\log k))$.
 Then for each such pattern $(\Tt,d)$ we test if it encodes some
 bull-free covering family. %
\Cref{lem:pattern-testing} allows us to do this in time $g(k,\, 16\cdot k)\cdot |V(G)|^3$.
We output the minimum of $\frac{\mathrm{odd}(G)+\mathrm{odd}(\Tt, d)}{2} +|\Tt|$ over all patterns  $(\Tt,d)$ that passed the test.

We prove the correctness of this algorithm.
Let $\alpha$ denote the computed value. 
\Cref{lem:pn-less-odd} implies that $\pn(G) \le \alpha$.
Let $\Pp$ be an optimal path partition of size $\pn(G)$ and $\Pp_4 \sub \Pp$ be the corresponding covering family.
By \Cref{lem:bull-free} the family $\Pp_4$ is bull-free.
By \Cref{lem:pattern-encoding} there exists a $(G,V_4,\ell)$-pattern $(\Tt,d)$ with $\ell \le 16\cdot k$ that encodes~$\Pp_4$.
Consequently, $(\Tt,d)$ will be considered by the algorithm and so $\alpha \le \frac{\mathrm{odd}(G)+\mathrm{odd}(\Tt, d)}{2} +|\Tt|$.
Finally, we apply \Cref{lem:pn-greater-odd} to obtain that 
$\frac{\mathrm{odd}(G)+\mathrm{odd}(\Tt, d)}{2} +|\Tt| \le |\Pp|$.
We infer that $\pn(G) \le \alpha \le |\Pp| = \pn(G)$. 
Hence, $\alpha = \pn(G)$.
This concludes the main proof.
\end{proof}

\section{Conclusion}\label{sec:conclusion}

We have presented an {FPT} algorithm to compute $\pn(G)$ parametrized by the edge-deletion distance of $G$ to a subcubic graph. 
A natural follow-up question is whether this parameter can be replaced by the vertex-deletion distance to a subcubic graph, which may be arbitrarily smaller than the edge-deletion distance.
Is this problem in FPT or in XP?
This question is non-trivial already for parameter value one: is there a polynomial algorithm that computes $\pn(G)$ of a graph $G$ with only one vertex of degree higher than three?

While the Gallai conjecture remains open,  Lov{\'a}sz proved that $n$ paths and cycles suffice to partition the edges of every graph on $2n$ vertices~\cite{Gallai}.
In the regime of ``parameterization above/below a guarantee'' one asks whether we can effectively compute a solution that is slightly better than a certain combinatorial bound (see, e.g.~\cite{fomin2025pathcover}).
The Lov{\'a}sz' bound gives rise to the following problem: given a graph on $2n$ vertices, determine whether $E(G)$ can be partitioned into $n-k$ path and cycles.
Is this problem in FPT or in XP for the parameter~$k$?
As before, the existence of a polynomial algorithm for $k=1$ is already open.

\bibliography{bib}

\newcommand{\etalchar}[1]{$^{#1}$}
\begin{thebibliography}{FMMR25}

\bibitem[BBB21]{Gallai-planar}
Alexandre Blanch{\'e}, Marthe Bonamy, and Nicolas Bonichon.
\newblock Gallai's path decomposition in planar graphs, 2021.
\newblock \href {https://arxiv.org/abs/2110.08870} {\path{arXiv:2110.08870}}.

\bibitem[BDG{\etalchar{+}}25]{baste2025polynomial}
Julien Baste, Lucas {De Meyer}, Ugo Giocanti, Etienne Objois, and Timothé
  Picavet.
\newblock A polynomial bound on the pathwidth of graphs edge-coverable by $k$
  shortest paths, 2025.
\newblock \href {https://arxiv.org/abs/2510.02901} {\path{arXiv:2510.02901}}.

\bibitem[BP19]{Bonamy19}
Marthe Bonamy and Thomas~J. Perrett.
\newblock Gallai’s path decomposition conjecture for graphs of small maximum
  degree.
\newblock {\em Discrete Mathematics}, 342(5):1293--1299, 2019.
\newblock \href {https://doi.org/10.1016/j.disc.2019.01.005}
  {\path{doi:10.1016/j.disc.2019.01.005}}.

\bibitem[CDF{\etalchar{+}}25]{chakraborty2025}
Dibyayan Chakraborty, Oscar Defrain, Florent Foucaud, Mathieu Mari, and
  Prafullkumar Tale.
\newblock Parameterized complexity of isometric path partition: treewidth and
  diameter, 2025.
\newblock \href {https://arxiv.org/abs/2508.05448} {\path{arXiv:2508.05448}}.

\bibitem[DFPT24]{Dumas24}
Ma\"{e}l Dumas, Florent Foucaud, Anthony Perez, and Ioan Todinca.
\newblock On graphs coverable by $k$ shortest paths.
\newblock {\em SIAM Journal on Discrete Mathematics}, 38(2):1840--1862, 2024.
\newblock \href {https://doi.org/10.1137/23M1564511}
  {\path{doi:10.1137/23M1564511}}.

\bibitem[DK00]{Gallai-2n-over-3-first}
Nathaniel Dean and Mekkia Kouider.
\newblock Gallai's conjecture for disconnected graphs.
\newblock {\em Discrete Mathematics}, 213(1--3):43--54, 2000.
\newblock \href {https://doi.org/10.1016/S0012-365X(99)00167-3}
  {\path{doi:10.1016/S0012-365X(99)00167-3}}.

\bibitem[FFM{\etalchar{+}}25]{Fernau25}
Henning Fernau, Florent Foucaud, Kevin Mann, Utkarsh Padariya, and Rajath {Rao
  K.N.}
\newblock Parameterizing path partitions.
\newblock {\em Theoretical Computer Science}, 1028:115029, 2025.
\newblock \href {https://doi.org/10.1016/j.tcs.2024.115029}
  {\path{doi:10.1016/j.tcs.2024.115029}}.

\bibitem[FGJ{\etalchar{+}}25]{fomin2025pathcover}
Fedor~V. Fomin, Petr~A. Golovach, Nikola Jedličková, Jan Kratochvíl, Danil
  Sagunov, and Kirill Simonov.
\newblock Path cover, {Hamiltonicity}, and independence number: An {FPT}
  perspective, 2025.
\newblock \href {https://arxiv.org/abs/2403.05943} {\path{arXiv:2403.05943}}.

\bibitem[FMMR25]{foucaud2025polynomialtime}
Florent Foucaud, Atrayee Majumder, Tobias Mömke, and Aida {Roshany-Tabrizi}.
\newblock Polynomial-time algorithms for path cover and path partition on trees
  and graphs of bounded treewidth, 2025.
\newblock \href {https://arxiv.org/abs/2511.07160} {\path{arXiv:2511.07160}}.

\bibitem[GST23]{GolovachST23}
Petr~A. Golovach, Giannos Stamoulis, and Dimitrios~M. Thilikos.
\newblock Model-checking for first-order logic with disjoint paths predicates
  in proper minor-closed graph classes.
\newblock In Nikhil Bansal and Viswanath Nagarajan, editors, {\em Proceedings
  of the 2023 {ACM-SIAM} Symposium on Discrete Algorithms, {SODA} 2023,
  Florence, Italy, January 22--25, 2023}, pages 3684--3699. {SIAM}, 2023.
\newblock \href {https://doi.org/10.1137/1.9781611977554.CH141}
  {\path{doi:10.1137/1.9781611977554.CH141}}.

\bibitem[KW10]{Kawarabayashi2010}
Ken-ichi Kawarabayashi and Paul Wollan.
\newblock A shorter proof of the graph minor algorithm: the unique linkage
  theorem.
\newblock In {\em Proceedings of the Forty-Second ACM Symposium on Theory of
  Computing}, STOC '10, pages 687--694, New York, NY, USA, 2010. Association
  for Computing Machinery.
\newblock \href {https://doi.org/10.1145/1806689.1806784}
  {\path{doi:10.1145/1806689.1806784}}.

\bibitem[Lov68]{Gallai}
L{\'a}szl{\'o} Lov{\'a}sz.
\newblock On covering of graphs.
\newblock In Paul Erd{\H{o}}s and Gyula O.~H. Katona, editors, {\em Theory of
  Graphs (Proceedings of the Colloquium held at Tihany, Hungary, September
  1966)}, pages 231--236. Academic Press, New York, 1968.

\bibitem[LY24]{Gallai-treewidth}
Changhong Lu and Niping Yi.
\newblock The problem of path decomposition for graphs with treewidth at most
  4.
\newblock {\em Discrete Mathematics}, 347(6):113957, 2024.
\newblock \href {https://doi.org/10.1016/j.disc.2024.113957}
  {\path{doi:10.1016/j.disc.2024.113957}}.

\bibitem[Man18]{manuel2018revisiting}
Paul Manuel.
\newblock Revisiting path-type covering and partitioning problems, 2018.
\newblock \href {https://arxiv.org/abs/1807.10613} {\path{arXiv:1807.10613}}.

\bibitem[MR18]{MajumdarR18}
Diptapriyo Majumdar and Venkatesh Raman.
\newblock Structural parameterizations of undirected feedback vertex set: {FPT}
  algorithms and kernelization.
\newblock {\em Algorithmica}, 80(9):2683--2724, 2018.
\newblock \href {https://doi.org/10.1007/S00453-018-0419-4}
  {\path{doi:10.1007/S00453-018-0419-4}}.

\bibitem[P{\'e}r84]{NP-complete-paper}
Bernard P{\'e}roche.
\newblock {NP}-completeness of some problems of partitioning and covering in
  graphs.
\newblock {\em Discrete Applied Mathematics}, 8(2):195--208, 1984.
\newblock \href {https://doi.org/10.1016/0166-218X(84)90101-X}
  {\path{doi:10.1016/0166-218X(84)90101-X}}.

\bibitem[SSS{\etalchar{+}}24]{Logic}
Nicole Schirrmacher, Sebastian Siebertz, Giannos Stamoulis, Dimitrios~M.
  Thilikos, and Alexandre Vigny.
\newblock Model checking disjoint-paths logic on topological-minor-free graph
  classes.
\newblock In Paweł Sobocinski, Ugo~Dal Lago, and Javier Esparza, editors, {\em
  Proceedings of the 39th Annual {ACM/IEEE} Symposium on Logic in Computer
  Science, {LICS} 2024, Tallinn, Estonia, July 8-11, 2024}, pages 68:1--68:12.
  {ACM}, 2024.
\newblock \href {https://doi.org/10.1145/3661814.3662089}
  {\path{doi:10.1145/3661814.3662089}}.

\bibitem[Yan98]{Gallai-2n-over-3-second}
Lirong Yan.
\newblock {\em On path decompositions of graphs}.
\newblock PhD thesis, Arizona State University, 1998.
\newblock Thesis (Ph.D.). ProQuest LLC, Ann Arbor, MI. MR2697319. ISBN:
  978-0591-80735-6.
\newblock URL:
  \url{https://www.proquest.com/openview/59f62a730dfd1ec06fefb62e49e2f83d/}.

\end{thebibliography}

\end{document}